\documentclass[10pt, twocolumn]{IEEEtran}
\normalsize
\usepackage{amsthm,amssymb,amsmath, tikz,graphics,cite,float,graphicx,epstopdf,epsfig,verbatim,url}
\usetikzlibrary{automata}
\usetikzlibrary{shapes,arrows}

\usepackage[utf8]{inputenc}

\usepackage{capt-of}
\usepackage[noend]{algpseudocode}
\usepackage{algorithmicx}
\usepackage[ruled]{algorithm}
\usepackage{enumerate}

\usepackage{lipsum}

\usepackage{mdwtab}
\usepackage{subfigure}
\usepackage{placeins}
\usepackage{setspace}

\allowdisplaybreaks

\theoremstyle{plain}
\theoremstyle{remark}

\begin{document}
\newtheorem{theorem}{Theorem}[section]
\newtheorem{lemma}{Lemma}
\newtheorem{conjecture}{Conjecture}
\newtheorem{corollary}{Corollary}
\newtheorem{definition}{Definition}
\newtheorem{property}{Property}
\newtheorem{remark}{Remark}

\bibliographystyle{plain}

\title{A Fair Power Allocation Approach to NOMA in Multi-user SISO Systems }

\author{Jos\'{e}~Armando~Oviedo,~\IEEEmembership{Student Member,~IEEE,}
        and~Hamid~R.~Sadjadpour,~\IEEEmembership{Senior Member,~IEEE}
\thanks{Copyright (c) 2017 IEEE. Personal use of this material is permitted. However, permission to use this material for any other purposes must be obtained from the IEEE by sending a request to pubs-permissions@ieee.org.}%
\thanks{The authors are with the Department of Electrical Engineering, University
of California, Santa Cruz, CA 95064 USA. (e-mail: \{xmando, hamid\}@soe.ucsc.edu.}}

\maketitle

\IEEEpeerreviewmaketitle

\begin{abstract}
	A non-orthogonal multiple access (NOMA) approach that always outperforms orthogonal multiple access (OMA) called Fair-NOMA is introduced. In Fair-NOMA, each mobile user is allocated its share of the transmit power such that its capacity is always greater than or equal to the capacity that can be achieved using OMA. For any slow-fading channel gains of the two users, the set of possible power allocation coefficients are derived. For the infimum and supremum of this set, the individual capacity gains and the sum-rate capacity gain are derived. It is shown that the ergodic sum-rate capacity gain approaches 1 b/s/Hz when the transmit power increases for the case when pairing two random users with i.i.d. channel gains. The outage probability of this approach is derived and shown to be better than OMA.

The Fair-NOMA approach is applied to the case of pairing a near base-station user and a cell-edge user and the ergodic capacity gap is derived as a function of total number of users in the cell at high SNR. This is then compared to the conventional case of fixed-power NOMA with user-pairing. Finally, Fair-NOMA is extended to $K$ users and it is proven that the capacity can always be improved for each user, while using less than the total transmit power required to achieve OMA capacities per user.

\end{abstract}

\section{Introduction}
\textit{Orthogonal multiple access} (OMA) is defined as a system that schedules multiple users in non-overlapping time slots or frequency bands during transmission. Therefore, if the signals for $k$ users, $k=1,\ldots, K$ are scheduled for transmission over a time period $T$, where $T$ is less than the coherence time of the channel, each user  transmits only $T/K$ amount of the total  transmission period (or fraction of the total bandwidth) with the entire transmit power $\xi$  allocated to that user.

\textit{Non-orthogonal multiple access} (NOMA) schedules transmission of the $K$ users’ signals simultaneously over the entire transmission period and bandwidth. Since the total transmit power $\xi $ must be shared between the $K$ users, a fraction $a_k\in(0,1)$ of the transmit power is allocated to user $k$, and $\sum_{k=1}^K a_k \leq 1$. In NOMA, each user employs successive interference cancellation (SIC) at the receiver to remove the interference of the signals from users with lesser channel gains \cite{InfTh:CT}.

An approach called Fair-NOMA is proposed for $K$ users in future wireless cellular networks. The underlying fundamental property of Fair-NOMA is that  users will always be guaranteed to achieve a capacity at least as good as OMA. The average capacities of two random users with i.i.d. channel SNR gains are derived along with the expected increase in capacity between OMA and NOMA. 



Another unique feature of our approach compared to the previous work is the fact that prior studies on NOMA have  focused on demonstrating that NOMA has advantages for increasing the capacity of the network when users are scheduled and paired based on their channel conditions (i.e. their location in the cell). Fair-NOMA does not rely on this condition since users' channel conditions are i.i.d. (i.e. location in the cell is not considered). Hence, all users will have equal opportunity to be scheduled, and thus is also completely "fair" from a time-sharing perspective. However, Fair-NOMA can be applied to any system with any scheduling and user-pairing approach.

The paper is organized as follows. Previous contributions on NOMA are discussed in section \ref{sec:previous}.
The assumptions and system model are outlined in section \ref{sec:system}. Section \ref{sec:fairNOMA} defines the Fair-NOMA power allocation region $\mathcal{A}_\text{FN}$, and develops its basic properties. The analysis of the effects of Fair-NOMA on the capacity of each user together with simulation results are provided in section \ref{sec:analysis}. The improvement in outage probability is derived and demonstrated in section \ref{sec:outage}. The application of Fair-NOMA to opportunistic user-pairing with near and cell-edge users is discussed and analyzed in section \ref{sec:fairNOMAMUD}. Section \ref{sec:MU-NOMA} defines Fair-NOMA for multi-user SISO systems. Finally, section \ref{sec:conclusion} concludes the paper and discusses  future work.

\section{Previous Work on NOMA} \label{sec:previous}

The concept of NOMA is based on using superposition coding (SC) at the transmitter and successive interference cancellation (SIC) at the receivers. This was shown to achieve the capacity of the channel by Cover and Thomas \cite{InfTh:CT}. The existence of a set of power allocation coefficients that allow all of the participating users to achieve capacity at least as good as OMA was suggested in \cite{FundWiCom:Tse}. 

Non-orthogonal access approaches using SC for future wireless cellular networks were mentioned in \cite{CompOMANOMA:WXP} as a way to increase single user rates when compared to CDMA. Schaepperle and Ruegg \cite{4GNOSig:SR} evaluated the performance of non-orthogonal signaling using SC and SIC in single antenna OFDMA systems using very little modifications to the existing standards, as well as how user pairing impacts the throughput of the system when the channel gains become increasingly disparate. This was then applied  \cite{WCSCMA:Schaep} to OFDMA wireless systems to evaluate the performance of cell edge user rates, proposing an algorithm that attempts to increase the average throughput and maintain fairness. These works do not assume to have the exact channel state information at the transmitter.

The concept of NOMA is  evaluated through simulation for full channel state information at the transmitter (CSIT) in the uplink \cite{ULNOMA:TakedaHiguchi} and downlink \cite{DLNOMA:TomidaHiguchi}, where the throughput of the system is shown to be on average always better than OMA when considering a fully defined cellular system evaluation, with both users occupying all of the bandwidth and time, and was compared to FDMA with each user being assigned an orthogonal channel. In \cite{SLDLNOMA:Saitoetal}, the downlink system performance throughput gains are evaluated by incorporating a complete simulation of an LTE cellular system (3GPP). 

Kim et. al. \cite{NOMABF:KKSK} developed an optimization problem that finds the power allocation coefficients for a broadcast MIMO NOMA system with $N$ base-station antennas serving $2N$ simultaneous users, where user-pairing is based on clustering two users with similar channel vector directions. Choi \cite{BFMIMONOMA:Choi} extends this work for a base-station with $L$ antennas and $K$ user-pairs ($L\geq2K$) by creating a two-stage beamforming approach with NOMA, such that closed-form solutions to the power allocation coefficients are found.

More recently, Sun et. al. explored the MIMO NOMA system ergodic capacity \cite{ErgMIMONOMA:SHIP} when the base-station has only statistical CSIT. Given that a near user has a larger expected channel gain than the cell-edge user, properties of the power allocation coefficients are derived, and a suboptimal algorithm to solving for the power allocation coefficients is proposed which maximizes the ergodic capacity. 


Fairness in NOMA systems is addressed in some works. The uplink case in OFDMA systems is addressed in \cite{5GNOMAUp:AXIT} by using an algorithm that attempts to maximize the sum throughput, with respect to OFDMA and power constraints. The fairness is not directly addressed in the problem formulation, but is evaluated using Jain's fairness index. In   \cite{PropFairNOMA:LMP}, a proportional fair scheduler and user pair power allocation scheme is used to achieve fairness in time and rate.  In \cite{FairnessNOMA5G:TK}, fairness is achieved in the max-min sense, where users are paired such that their channel conditions are not too disparate, while the power allocation maximizes the rates for the paired users. 

Ding et. al. \cite{5GNOMA:DFP} provide an analysis for fixed-power NOMA (F-NOMA) and cognative radio NOMA (CR-NOMA). In F-NOMA, with a cell that has $N$ total users, it is shown that the probability that NOMA outperforms OMA asymptotically approaches $1$. In CR-NOMA, a primary user is allowed all of the time and bandwidth, unless an opportunistic secondary user exists with a stronger channel condition relative to the primary user, such that transmitting both of their signals will not reduce the primary user's SINR below some given threshold. It is shown that the diversity order of the $n$-th user is equal to the order of the weaker $m$-th user, leading to the conclusion that this approach benefits from pairing the two users with the strongest channels.

The main contribution of this work is to demonstrate that NOMA capacity can fundamentally always outperform OMA capacity for each user, regardless of the channel conditions of the users, and to derive exactly what the power allocation should be for each user to achieve this, based on their channel gains. Furthermore, the expected sum-rate capacity gain made when the users are paired with i.i.d. random channels is 1 bps/Hz at high SNR, even in the extreme case when all of the transmit power is allocated to the stronger user. 
Furthermore, the outage probabilities are also derived for this case, and shown to decrease for each user, but significantly for the weaker user. The approximate sum-rate capacity gain for the case of pairing the strongest and weakest users in the cell is derived for the high SNR regime, and compared to user-pairing with fixed-power approaches. Lastly, the more general case of $K$ user SISO NOMA is considered, and it is fundamentally proven that a NOMA power allocation strategy always exists that achieves equal or greater capacity per user when compared to OMA.



\section{System Model and Capacity} \label{sec:system}

Let a mobile user $i$ have a signal $x_i$ transmitted from a single antenna base-station (BS). The channel gain is $h_i\in\mathbb{C}$ with SNR gain p.d.f. $f_{|h|^2}(w) = \frac{1}{\beta}e^{-\frac{w}{\beta}}$, and receiver noise is complex-normal distributed $z_i\sim\mathcal{CN}(0,1)$. In the two-user case, if user-1 and user-2 transmit their signals half of period $T$ utilizing the entire transmit power $\xi$, 
then the received signal for each user is $y_i = h_i \sqrt{\xi}  x_i + z_i, i=1,2$. If $\mathbb{E}[ |x_i|^2 ] = 1$, the information capacity of each user is  
	$C_i^\text{O} = \frac{1}{2}\log_2\left(1 + \xi |h_i|^2\right).$ 
The sum-rate capacity for OMA is therefore $S_\text{O} = C_1^\text{O}+C_2^\text{O}$. For the case of NOMA, it is assumed that user-2 channel gain is the larger one ($|h_2|^2 > |h_1|^2$), then user-2  can perform SIC at the receiver by  treating its own signal as noise and decoding user-1's signal first. If the power allocation coefficient for user-2 is $a\in(0,1/2)$, then user-1's signal is allocated $1-a$ transmit power, and the received signals for both users are
\begin{align}
	y_i &= \sqrt{(1-a)\xi }h_i x_1 + \sqrt{a\xi }h_i x_2 + z_i,  i=1,2.
\end{align}
Since $|h_2|^2>|h_1|^2$, it follows that 
	$\frac{(1-a)\xi |h_2|^2}{a\xi |h_2|^2+1} > \frac{(1-a)\xi |h_1|^2}{a\xi |h_1|^2+1},$
which allows user-2's receiver to perform SIC and remove the interference from user-1's signal. Hence, the capacity for each user is
\begin{align}
		&C_1^\text{N}(a) = \log_2\left( 1 + \dfrac{(1-a)\xi |h_1|^2}{a\xi |h_1|^2 + 1}\right), \\
		&C_2^\text{N}(a) = \log_2\left( 1 + a\xi |h_2|^2\right).
\end{align}
The sum-rate capacity for NOMA is therefore $S_\text{N}(a) = C_1^\text{N}(a)+C_2^\text{N}(a)$. These capacity expressions are used in each case of OMA and NOMA to find the values of $a$ that make NOMA "fair."

\section{Fair-NOMA Power Allocation Region} \label{sec:fairNOMA}
In order for user-1  NOMA capacity to be greater than or equal to OMA capacity, it must be true that $C_1^\text{N}(a)\geq C_1^\text{O}$. Solving this inequality for $a$ gives $a \leq \frac{\sqrt{1 + \xi |h_1|^2} - 1}{\xi |h_1|^2}$. Similarly, for user-2 when $C_2^\text{N}(a)\geq C_2^\text{O}$ results in $a \geq \frac{\sqrt{1 + \xi |h_2|^2} - 1}{\xi |h_2|^2}$. Both the upper and lower bounds on the transmit power fraction $a$ to achieve better sum and individual capacities have the form $a(x) = (\sqrt{1 + \xi x}-1)/(\xi x)$.

Define 
\begin{equation}
	\begin{array}{ccc} a_{\inf} = \dfrac{\sqrt{1 + \xi |h_2|^2} - 1}{ \xi |h_2|^2}$  and $a_{\sup} = \dfrac{\sqrt{1 + \xi |h_1|^2} - 1}{\xi |h_1|^2}. \end{array}
\end{equation} 
Then by Property 1 in \cite{FairNOMAInfocom2016}, it is clear that if $|h_2|^2>|h_1|^2 \Rightarrow a_\text{sup} > a_\text{inf}$. The Fair-NOMA power allocation region is therefore defined as $\mathcal{A}_\text{FN}=[a_\text{inf}, a_\text{sup}]$, and selecting any $a\in\mathcal{A}_\text{FN}$ gives
	$ C_1^\text{N}(a)\geq C_1^\text{O},$
	$C_2^\text{N}(a)\geq C_2^\text{O},$ and 
	$S_\text{N}(a) > S_\text{O}$.
Since the sum-rate capacity $S_\text{N}(a)$ is a monotonically increasing function of $a$, then $a_\text{sup} = \arg{\displaystyle\max_{a\in\mathcal{A}_\text{FN}}}(C_2^\text{N}(a))$ also maximizes $S_\text{N}(a)$. The sum-rate capacity of NOMA is strictly larger than the sum-rate capacity of OMA because at the least one of the user's capacities always increases.

\begin{theorem}\label{thm:thm1}
	For a two-user NOMA system that allocates power fraction $1-a$ to user-1 and $a$ to user-2, such that $a\in\mathcal{A}_\text{FN}$, the sum-rate $S_\text{N}(a)$ is a monotonically increasing function of both $|h_1|^2$ and $|h_2|^2$. 
\end{theorem}
\begin{proof}
	See appendix \ref{proof:thm1}.
\end{proof}
This result implies that as the channel gain for the weaker user increases, the total capacity increases while the power allocation to the stronger user decreases. 
This means that, as the channel gain $|h_1|^2$ increases towards the value of $|h_2|^2$, then the capacity gain by user-1 is greater than the capacity loss by user-2. In the extreme case where $ |h_1|^2 \rightarrow |h_2|^2$, then $a_\text{sup}\rightarrow a_\text{inf}$, and both $C_1^\text{N}(a)$ and $C_2^\text{N}(a)\rightarrow C_2^\text{O}$.  In other words, the Fair-NOMA capacity is upper bounded by the capacity obtained by allocating all of the transmit power to the stronger user.
This is somewhat related to the multiuser diversity concept result in \cite{knopp95} for OMA systems, which suggests allocating all the transmit power to the stronger users will increase the overall capacity of the network.

In contrast, with the increase in $|h_2|^2$,  $C_2^\text{N}(a_\text{sup})$ increases and hence the capacity gains from Fair-NOMA increase. Therefore with Fair-NOMA, as is the same with the previously obtained result for fixed-power allocation NOMA, when $|h_2|^2 - |h_1|^2$ increases, $S_\text{N}(a)-S_\text{O}$ also increases \cite{InfTh:CT,5GNOMA:DFP}. 
This will be further exemplified in Section \ref{sec:fairNOMAMUD}, Theorem \ref{thm:NOMAminmax}. 

\vspace{-0.05in}
\section{Analysis of Fair-NOMA Capacity}\label{sec:analysis}

\subsection{Expected Value of Fair-NOMA Capacity}

The expected value of the Fair-NOMA capacities of the two users depend on the power allocation coefficient $a$. In order to determine the bounds of this region, the expected value of capacity of each user is derived for the cases of $a_\text{inf}$ and $a_\text{sup}$ and compared with that of OMA.

Since the channels of the two users are i.i.d. random variables, let the two users selected have channel SNR gains of $|h_i|^2$ and $|h_j|^2$, where $f_{|h|^2}(x)=\frac{1}{\beta}e^{-\frac{x}{\beta}}$. Since we call the user with weaker (stronger) channel gain user-$1$ (user-$2$), then $|h_1|=\min\{|h_i|^2, |h_j|^2\}$ and $|h_2|^2=\max\{|h_i|^2, |h_j|^2\}$. Therefore, the joint pdf of $|h_1|^2$ and $|h_2|^2$ is 
\begin{equation}
	f_{|h_1|^2, |h_2|^2}(x_1,x_2) = \frac{2}{\beta^2}e^{-\frac{x_1+x_2}{\beta}}.
\end{equation}
It is shown \cite{FairNOMAInfocom2016} that the ergodic capacities and the sum-rate of users in OMA are
\begin{align}
	&\mathbb{E}[ C_1^\text{O} ]  = \frac{e^{\frac{2}{\beta\xi}}}{\ln(4)} E_1\left(\frac{2}{\beta\xi}\right),\\
	&\mathbb{E}[ C_2^\text{O} ] = \frac{e^{\frac{1}{\beta\xi}}}{\ln(2)}E_1\left(\frac{1}{\beta\xi}\right) - \frac{e^{\frac{2}{\beta\xi}}}{\ln(4)} E_1\left(\frac{2}{\beta\xi}\right),\\
	&\mathbb{E}[S_{\text{O}} ] = \frac{e^{\frac{1}{\beta\xi}}}{\ln(2)}E_1\left(\frac{1}{\beta\xi}\right)
\end{align}
where $E_1(x) =\int_x^\infty u^{-1}e^{-u} du$ is the well-known exponential integral. Note that $\mathbb{E}[ C_1^\text{O} ]=\mathbb{E}[ C_1^\text{N}(a_\text{sup}) ]$ and $\mathbb{E}[ C_2^\text{O} ] = \mathbb{E}[ C_2^\text{N}(a_\text{inf}) ]$.




It is also shown \cite{FairNOMAInfocom2016} that
\begin{align}
\label{eq:C1ainf}  &\mathbb{E}\left[ C_1^\text{N}(a_\text{inf}) \right] = \frac{3e^{\frac{2}{\beta\xi}}}{\ln(4)}E_1\left(\frac{2}{\beta\xi}\right) \\
	&- \int_0^\infty \frac{2}{\beta\ln(2)} \cdot \exp\left(-\frac{x}{\beta}\left( \frac{\sqrt{1+\xi x}-2}{\sqrt{1+\xi x}-1}\right)\right)\nonumber\\
    & \times \left( E_1\left( \frac{x}{\beta(\sqrt{1+\xi x}-1)} \right) - E_1\left(\frac{x\sqrt{1+\xi x}}{\beta(\sqrt{1+\xi x}-1)} \right) \right)  dx, \nonumber
\end{align}
and
\begin{align}\label{eq:C2asup}
	&\mathbb{E}[ C_2^\text{N}(a_\text{sup}) ]  = \frac{e^{\frac{2}{\beta\xi}}}{\ln(4)} E_1\left(\frac{2}{\beta\xi}\right) + \int_0^\infty \frac{2}{\beta\ln(2)}  \\ 
	 &\times \exp\left(-\frac{x}{\beta}\left(\frac{\sqrt{1+\xi x} -2}{\sqrt{1+\xi x} -1}\right)\right) E_1\left( \frac{x\sqrt{1+\xi x}}{\beta(\sqrt{1+\xi x} - 1)} \right)dx.\nonumber
\end{align}

At high SNR ($\xi\gg 1$), the approximate capacities are
\begin{align}
 	& C_i^\text{O} \approx  \frac{1}{2}\log_2(\xi|h_i|^2), \\
	& C_1^\text{N}(a_\text{inf}) \approx  \frac{1}{2}\log_2(\xi|h_2|^2), \\
	& C_2^\text{N}(a_\text{sup})\approx \log_2\left(\sqrt{\frac{\xi}{|h_1|^2}}|h_2|^2\right).
\end{align}
This implies that when $\xi\gg 1$, $C_1^\text{N}(a_\text{sup}) \approx C_2^\text{O}$.  
The high SNR approximations lead to following result for the difference in the expected capacity gains, i.e.,  $\Delta S(a) = S_\text{N}(a)-S_\text{O}$. 

\begin{theorem}\label{thm:DiffCap1}
    In a two-user SISO system with $|h_i|^2\sim\mathrm{Exponential}(\frac{1}{\beta})$ and at high SNR regime, the increase in sum capacity is $\mathbb{E}[\Delta S(a)]\approx 1$ bps/Hz, $\forall a\in\mathcal{A}_\text{FN}$. 
\end{theorem}
\begin{IEEEproof}
    See appendix \ref{proof:DiffCap1}. 
\end{IEEEproof}
This interesting result means that when the transmit power approaches infinity, the average increase in sum capacity is  the same for both $a_\text{inf}$ and $a_\text{sup}$ and is equal to $1\text{ bps/Hz}$. Equivalently, it means that $\forall a\in\mathcal{A}_\text{FN}$, both users experience an expected increase in capacity over OMA of $c$ and $1-c$ where $c\in[0,1]$.

\subsection{Comparison of Theoretical and Simulations Results}

Fair-NOMA theoretical results are compared with simulation when  $\beta=1$. 
In figure \ref{fig:oma}, the capacity of NOMA is compared with that of OMA for both users, including the high SNR approximations. As can be seen, the theoretical derivations match the simulation results. The performances of $C_1^\text{N}(a_\text{inf})$ and $C_2^\text{N}(a_\text{sup})$ are plotted.
The simulation of $\mathbb{E}[C_1^\text{N}(a_\text{inf})]$ matches the theoretical result in equation (\ref{eq:C1ainf}), and the simulation of $\mathbb{E}[C_2^\text{N}(a_\text{sup})]$ matches the theoretical result in equation (\ref{eq:C2asup}). The high SNR approximations show to be very close for values of $\xi>25$ dB.
Since $C_2^\text{O}=C_2^\text{N}(a_\text{inf})$ and $C_1^\text{O}=C_1^\text{N}(a_\text{sup})$, it is apparent from the plots that the gain in performance is always approximately $1\text{ bps/Hz}$ for one of the users and also the sum capacity when using Fair-NOMA \cite{FairNOMAInfocom2016}. 
\begin{figure}[http]
\vspace{-5mm}

    \centering
      \includegraphics[width=\columnwidth]{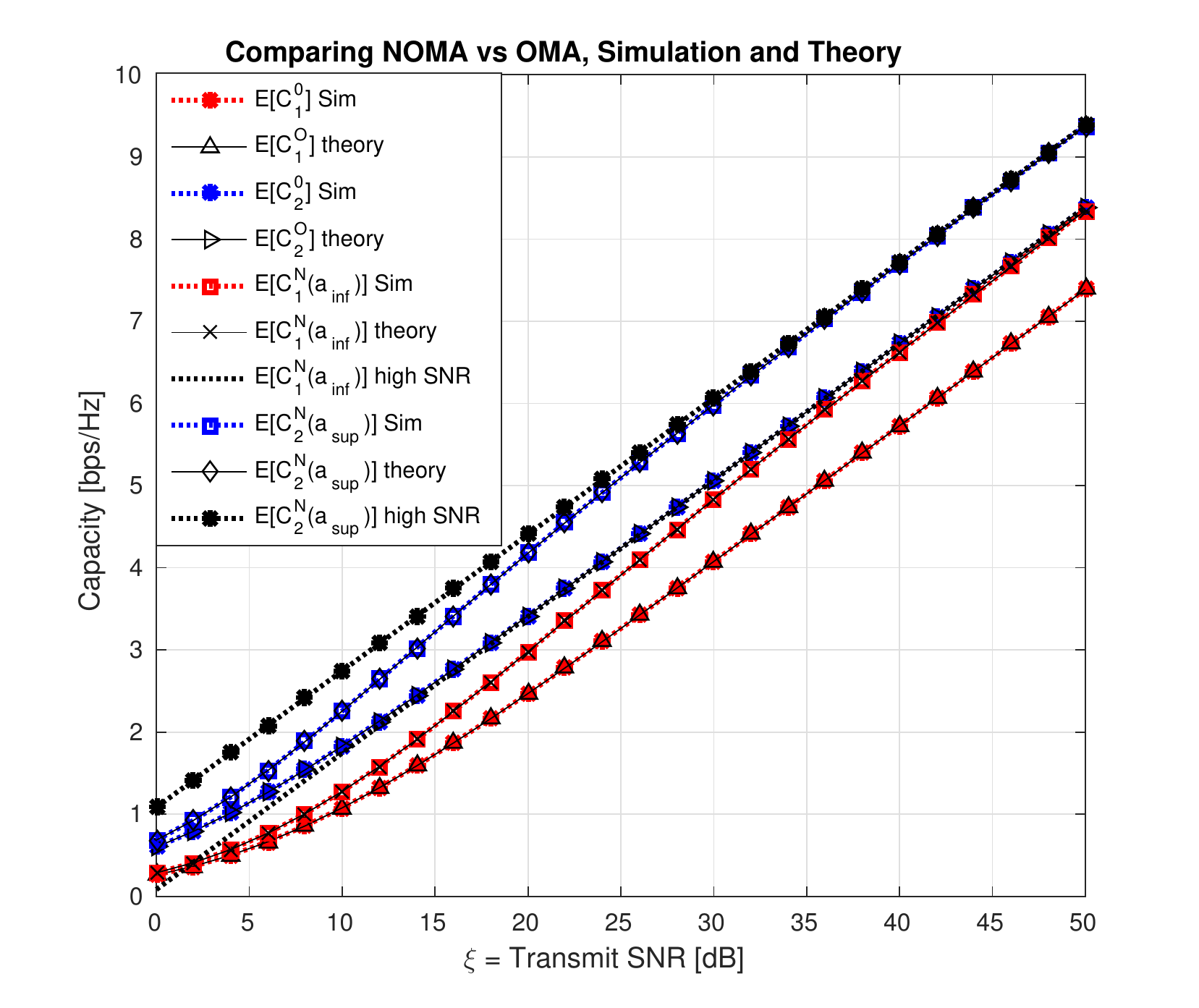}
\vspace{-5mm}

\caption{\label{fig:oma}Comparing the capacity of NOMA and OMA}
\end{figure}

\section{Outage Probability of Fair-NOMA}\label{sec:outage}

Suppose that the minimum rate that is allowed by the system to transmit a signal is $R_0$. The probability that a user cannot achieve this rate with any coding scheme is given by $\mathrm{Pr}\{C < R_0\}$. As with the average capacity analysis, the outage performance of NOMA is analyzed by looking at $a_\text{inf}$ and $a_\text{sup}$, and then draw logical conclusions from that. 

The outage probability of user-1 using  OMA is given by
\begin{align}
	p_{\text{O},1}^\text{out}&=\mathrm{Pr}\left\{\log_2(1+\xi|h_1|^2)^{1/2} < R_0\right\}\\
	 &= \int_0^{\frac{4^{R_0}-1}{\xi}}\int_{x_1}^\infty \frac{2}{\beta^2}e^{-\frac{x_1+x_2}{\beta}}dx_2dx_1 \\
	&= 1-\exp\left(-\tfrac{2(4^{R_0}-1)}{\beta\xi}\right)\nonumber.
\end{align}
For user-2 using OMA, the outage probability is given by
\begin{align}
	p_{\text{O},2}^\text{out}=&\mathrm{Pr}\left\{\log_2(1+\xi|h_2|^2)^{1/2} < R_0\right\}\\
	=& \int_0^{\frac{4^{R_0}-1}{\xi}}\int_{x_1}^{\frac{4^{R_0}-1}{\xi}} \frac{2}{\beta^2}e^{-\frac{x_1+x_2}{\beta}}dx_2dx_1, \\
	=& 1 + \exp\left(-\tfrac{2(4^{R_0}-1)}{\beta\xi}\right) - 2\exp\left(-\tfrac{4^{R_0}-1}{\beta\xi}\right).\nonumber
\end{align} 

Denote the NOMA outage probability for user-$i$ as $p_{\text{N},i}^\text{out}(a)$ such that 
	$p_{\text{N},i}^\text{out}(a) = \mathrm{Pr}\{C_i^\text{N}(a) < R_0 \}$ for $i=1,2$.
It should be obvious that $p_{\text{N},1}^\text{out}(a_\text{sup})=p_{\text{O},1}^\text{out}$ and $p_{\text{N},2}^\text{out}(a_\text{inf})=p_{\text{O},2}^\text{out}$. The outage probabilities $p_{\text{N},1}^\text{out}(a_\text{inf})$ and $p_{\text{N},2}^\text{out}(a_\text{sup})$ are provided in the following property.
\begin{property} \label{prop:outage} Outage Probabilities $p_{\text{N},1}^\text{out}(a_\text{inf})$ and $p_{\text{N},2}^\text{out}(a_\text{sup})$:
\begin{enumerate}[(a)]
	\item \label{item:outage1}
	The outage probability for user-1 at $a=a_\text{inf}$ is given by 
	\begin{align}
		\label{eq:NOMAoutMU1} p_{\text{N},1}^\text{out}(a_\text{inf}) = 1 + e^{-\frac{\alpha_2}{\beta}} - \frac{2}{\beta}\int_{\alpha_2}^\infty e^{-\frac{x(\alpha_1+1)}{\beta}}dx,
	\end{align} 
	where $\alpha_1$ and $\alpha_2$ are defined as 
	\begin{align*}
		 &\alpha_1= \tfrac{2^{R_0}-1}{\xi x + 2^{R_0}(1-\sqrt{1+\xi x}) }, \\
		 &\alpha_2= \tfrac{4^{R_0}-2}{2\xi} + \sqrt{ \tfrac{4^{R_0}-1}{\xi^2}+\tfrac{(4^{R_0}-2)^2}{4\xi^2} }.
	\end{align*}
	\item \label{item:outage2} 
The outage probability for user-2 at $a=a_\text{sup}$ is given by   
\end{enumerate}
	\begin{align}
	\hspace{-0.2in}	p_{\text{N},2}^\text{out}(a_\text{sup}) =  1 + e^{-\frac{2(4^{R_0}-1)}{\beta\xi}} - 2e^{-\frac{2(2^{R_0}-1)}{\beta\xi}} + (2^{R_0}-1)\nonumber\\
	e^{\frac{(2^{R_0}-3)^2}{4\beta\xi} }	\cdot\sqrt{\tfrac{\pi}{\beta\xi}}\left[\mathrm{erfc}\left(\tfrac{2^{R_0}+1}{2\sqrt{\beta\xi}}\right) - \mathrm{erfc}\left(\tfrac{3(2^{R_0})-1}{2\sqrt{\beta\xi}}\right)\right].
		\label{eq:outage2}
	\end{align} 
\end{property} 
\begin{IEEEproof}
	See appendix \ref{proof:Outage_Probabilities}.
\end{IEEEproof}
There is no closed form solution for the integral in $p_{\text{N},1}^\text{out}(a_\text{inf})$, however it can be easily computed by a computer.


Figure \ref{fig:nomaout} plots the outage probabilites of OMA and NOMA for different values of $a$ and for $R_0=2$ bps/Hz. The probability of user-1 experiencing an outage is clearly greater than for user-2. However, the reduction of the outage probability for user-1 using $a=a_\text{inf}$ becomes significant as $\xi$ increases, to the effect of nearly 1 order of magnitude drop-off when $\xi$ is really large. The outage probability reduction for user-2 is not as significant as the improvement made by user-1. However,  when $a=a_\text{sup}$, the same outage probability can be obtained using NOMA with  $\xi$ approximately 2 dB less than is required when using OMA. Thus, even when the power allocation coefficient $a$ is restricted to being in $\mathcal{A}_\text{FN}$, the probability of users to be able to achieve their minimum service requirement rates $R_0$ is improved, and especially improved for the weaker channel  gain. Even when the power allocation coefficient $a=(a_\text{inf}+a_\text{sup})/2$, the outage probabilities of both users improves significantly when using NOMA. 
\begin{figure}
	\centering
	\includegraphics[width=\columnwidth]{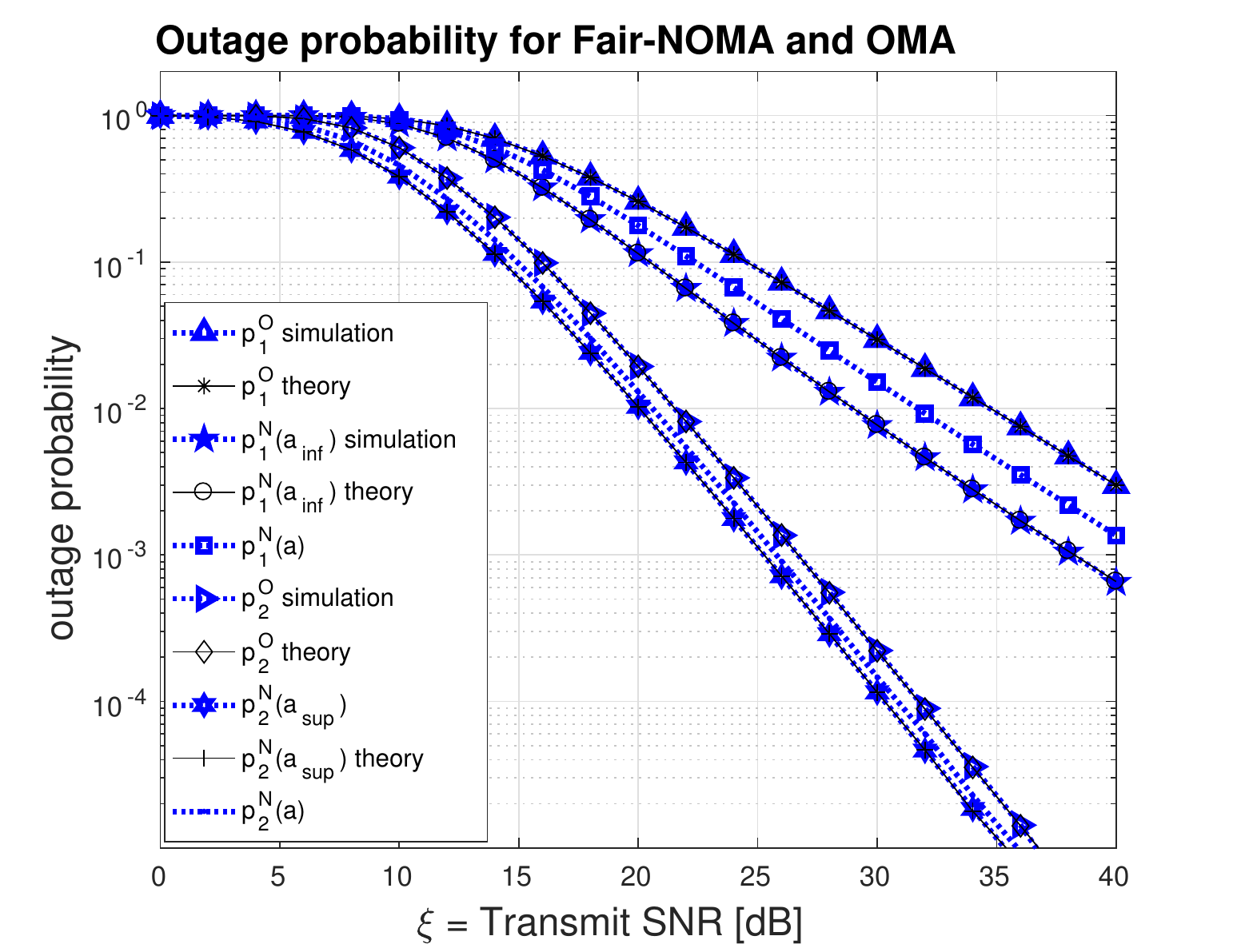}
\vspace{-5mm}
	
	\caption{\label{fig:nomaout}Outage probabilities of NOMA and OMA as functions of $\xi$.}
\end{figure}



\vspace{-1mm}
\section{Fair-NOMA in Opportunistic User-Pairing}\label{sec:fairNOMAMUD}

It has been suggested in \cite{5GNOMA:DFP} that the best NOMA performance is obtained when user channel conditions are most disparate, i.e. pairing the user with the weakest channel condition and the user with the strongest channel condition together. However, it is not known what the expected capacity gap is in this case, particularly for the case when both users are allocated power such that they both always outperform their OMA performance. Since the power allocation scheme where NOMA outperforms OMA with probability of 1 has been defined, regardless of number of users, this approach can also be applied here. 
 
Suppose there exists a set of $K$ mobile users in a cell, and two of these users can be scheduled during the same transmission period. It is of particular interest to select the users that have the largest difference in channel SNR gain. If the channel SNR gains of the users are i.i.d. $\mathrm{Exponential}(\frac{1}{\beta})$, and the two selected users have the minimum and maximum channel SNR gains, how much of an improvement in the sum-rate capacity will be observed by using NOMA versus OMA? 

\vspace{-0.15in}
\subsection{Analysis of Fair-NOMA with Opportunistic User-Pairing}
Let $|h_0|^2 = \min(|h_1|^2, \ldots,|h_K|^2)$ and $|h_M|^2 = \max(|h_1|^2, \ldots,|h_K|^2)$. In order to compute the expected sum-rate capacity,  the joint CDF $F_{|h_0|^2,|h_M|^2}(x_0, x_M)$ and PDF of $f_{|h_0|^2,|h_M|^2}(x_0,x_M)$ are needed. It is easily shown that
\begin{align}
	&\mathrm{Pr}\{|h_M|^2<x_M\} = \mathrm{Pr}\{|h_0|^2<x_0, |h_M|^2<x_M \} \nonumber\\
			& + \mathrm{Pr}\{|h_0|^2>x_0, |h_M|^2<x_M \},\\
	\Rightarrow &F_{|h_0|^2,|h_M|^2}(x_0, x_M) = \mathrm{Pr}\{|h_0|^2<x_0, |h_M|^2<x_M \} \\
			\label{eq:JointProb2}&= \mathrm{Pr}\{|h_M|^2<x_M\} - \mathrm{Pr}\{|h_0|^2>x_0, |h_M|^2<x_M \}.
\end{align}
The first term on the right in equation (\ref{eq:JointProb2}) is the CDF of the maximum of $K$ i.i.d. exponential random variables, which is given by 
\begin{equation}
	\mathrm{Pr}\{|h_M|^2<x_M\} = (1-e^{-\frac{x_M}{\beta}})^K.
\end{equation}
The second term can be easily computed.
\begin{align*}
	\mathrm{Pr}\{|h_0|^2>x_0, |h_M|^2<x_M \} =& \int_{x_0}^{x_M}\hspace{-2mm}\cdots\int_{x_0}^{x_M} \prod_{k=1}^{K}\frac{e^{-\frac{x_k}{\beta}}}{\beta}dx_k \\
	=& (e^{-\frac{x_0}{\beta}} - e^{-\frac{x_M}{\beta}})^K
\end{align*}
Therefore, the joint CDF is given by
\begin{align}
	F_{|h_0|^2,|h_M|^2}(x_0, x_M) =& (1 - e^{-\frac{x_M}{\beta}})^K-(e^{-\frac{x_0}{\beta}} - e^{-\frac{x_M}{\beta}})^K,
\end{align}
and the joint PDF is
\begin{align}
	f_{|h_0|^2,|h_M|^2}&(x_0, x_M) \\
	=&\frac{K(K-1)}{\beta^2}e^{-\frac{x_0+x_M}{\beta}}(e^{-\frac{x_0}{\beta}}-e^{-\frac{x_M}{\beta}})^{K-2}.\nonumber
\end{align}
The following theorem provides the sum-rate capacity increase of NOMA when $\xi|h_0|^2\gg 1$.
\begin{theorem}\label{thm:NOMAminmax}
	Let $\{|h_1|^2, \ldots,|h_K|^2\}$ be the i.i.d. SISO channel SNR gains of $K$ users, such that the two users selected for transmission together have the minimum and maximum channel SNR gains. When $\xi\ |h_0|^2 \gg 1$, the sum-rate capacity increase from OMA to NOMA for $a = a_\text{sup}$ is
	\begin{equation}\label{eq:NOMAMUDDiff}
		\mathbb{E}[\Delta S(a_\text{sup})] \approx \frac{1}{2}\log_2(K) + \frac{1}{2}\sum_{m=2}^K\binom{K}{m}(-1)^{m} \log_2 (m) .
	\end{equation}
\end{theorem}
\begin{IEEEproof}
	See appendix \ref{proof:NOMAminmax}.
\end{IEEEproof}
\begin{remark}
	This result is similar to the result obtained for the 2-by-2 MIMO case in Lemma  2, equation 33 in \cite{MIMONOMA:DAP}, except a fixed-power allocation approach was used there, whereas the result above uses a Fair-NOMA power allocation approach. Although the expected capacity gap $\mathbb{E}[\Delta S(a)]$ increases when selecting $a=a_\text{sup}$, caution should be used when utilizing the fixed-power approach to not set $a$ too close to the value of $a_\text{inf}$. An approximation of the capacity gap using $\mathbb{E}[\Delta S(a_\text{inf})]$ for large $\xi$ and $K$ is given as 
	\begin{align*}
		&\mathbb{E}[\Delta S(a_\text{inf})]\approx \frac{e^{\frac{K}{\beta\xi}}}{\ln(4)}E_1\left(\frac{K}{\beta\xi}\right)\\ 
		&- \log_2\left(1+\frac{\sqrt{1+\xi(\psi(K+1)+\gamma)}-1)}{K(\psi(K+1)+\gamma)}\right).
	\end{align*}
	where $\psi(w)=\Gamma'(w)/\Gamma(w)$ is the digamma function, $\Gamma(w)=\int_0^\infty u^{w-1}e^{-u}du$ is the gamma function, and $\gamma=-\int_0^\infty e^{-u}\ln(u)du$ is the Euler-Mascheroni constant. It can be seen in figure \ref{fig:NOMAmud_deltaS} that as the number of users increases, the expected capacity gap actually decreases. Therefore, even for fixed-power allocation approaches to NOMA, $a$ should be selected to be greater than $a_\text{inf}$ for the case of pairing minimum and maximum channel gain users. 
\end{remark}
This result shows that the sum-rate capacity difference increases as a function of $K$. However, this increase is slow. Nonetheless, there is a fundamental limit to the amount the capacity can increase when using Fair-NOMA, while maintaining the capacity of the weaker user equal to the capacity using OMA. 

It is important to note that  as the number of mobile users becomes very large, while pairing the strongest and weakest users together will give us the greatest increase in sum-rate capacity, it does not maximize sum-rate capacity itself. This can be seen from theorem \ref{thm:thm1}, which states that the sum-rate capacity actually increases as the channel gain of the weaker user monotonically increases. A practical way of viewing this issue is that, as the number of users $K$ increases, the weakest user has channel gain that in probability is too weak to achieve the quality of service threshold rate $R_0$.Should no outage rate be specified, the weaker user achieves such a low capacity, that the stronger user contributes most of the capacity, while using nearly half the transmit power, according to Property 1 from \cite{FairNOMAInfocom2016}. Hence, a little more than half of the transmit power is nearly wasted.

\subsection{Comparing Simulation Results with Analysis}

For the simulation results, the performance of Fair-NOMA  combined with opportunistic user-pairing is compared to the performance of  OMA and  fixed-power NOMA. The simulations are run for different values of $\xi$ and $K$. For fixed-power NOMA, the power allocation coefficient is a constant value of $a=\frac{1}{5}$, such that the weaker user is  allocated $\frac{4}{5}$ of the transmit power.

Figure \ref{fig:NOMAmud_capacity} shows the average capacities of both the weakest and strongest users versus $K$  and for each case of $a=a_\text{inf}$ and $a_\text{sup}$. The capacity of the stronger user is shown to exhibit the effects of multiuser diversity, since not only does its channel gain grow as $K$ increases, but also the power allocated also increases when $a=a_\text{sup}$, thus providing the increase in capacity predicted in equation (\ref{eq:NOMAMUDDiff}). In the case of $a=a_\text{inf}$ the capacity is initially shown to increase as $K$ increases, due to $a_\text{inf}$ decreasing with $|h_M|^2$ according to Property 1 from \cite{FairNOMAInfocom2016}. However, as $K$ continues to increase, the weakest users capacity eventually begins to decrease due to its channel gain being the minimum of a large number of users, and thus this term begins to dominate the capacity behavior. 
\begin{figure}[http]
	\centering
	\includegraphics[width=\columnwidth]{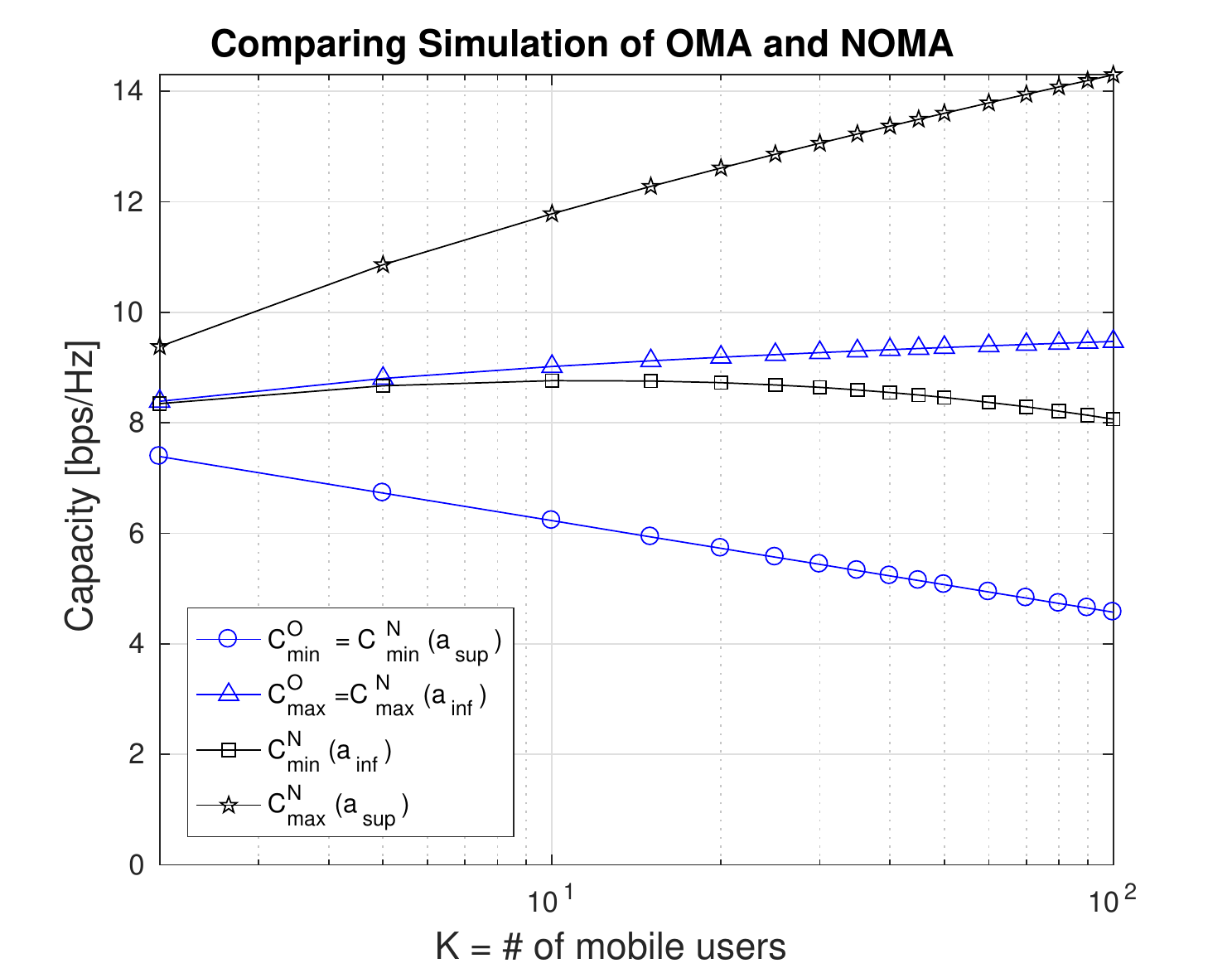}
\vspace{-3mm}
	\caption{\label{fig:NOMAmud_capacity} Ergodic capacity with opportunistic user-pairing, $\xi=50$ dB}
\end{figure}

The sum-rate capacity for Fair-NOMA with $a=a_\text{sup}$, fixed-power NOMA with $a=\frac{1}{5}$, and OMA are shown in Fig. \ref{fig:NOMAmud_sumcompare}. As expected, the sum-rate capacity for each user at lower values of $\xi$ performs best when applying Fair-NOMA when compared to fixed-power NOMA. This is because Fair-NOMA always guarantees a capacity increase, i.e. with probability 1, while fixed-power NOMA only achieves higher capacity with probability as given in \cite{5GNOMA:DFP}. However, as $\xi$ increases, both  capacities of Fair-NOMA and fixed-power NOMA approach the same value asymptotically. This agrees with the result obtained that at high SNR, the capacity gain should reach a limit when $\xi\rightarrow\infty$, no matter how much extra power is allocated to the stronger user. 
\begin{figure}[http]
	\centering
	\includegraphics[width=\columnwidth]{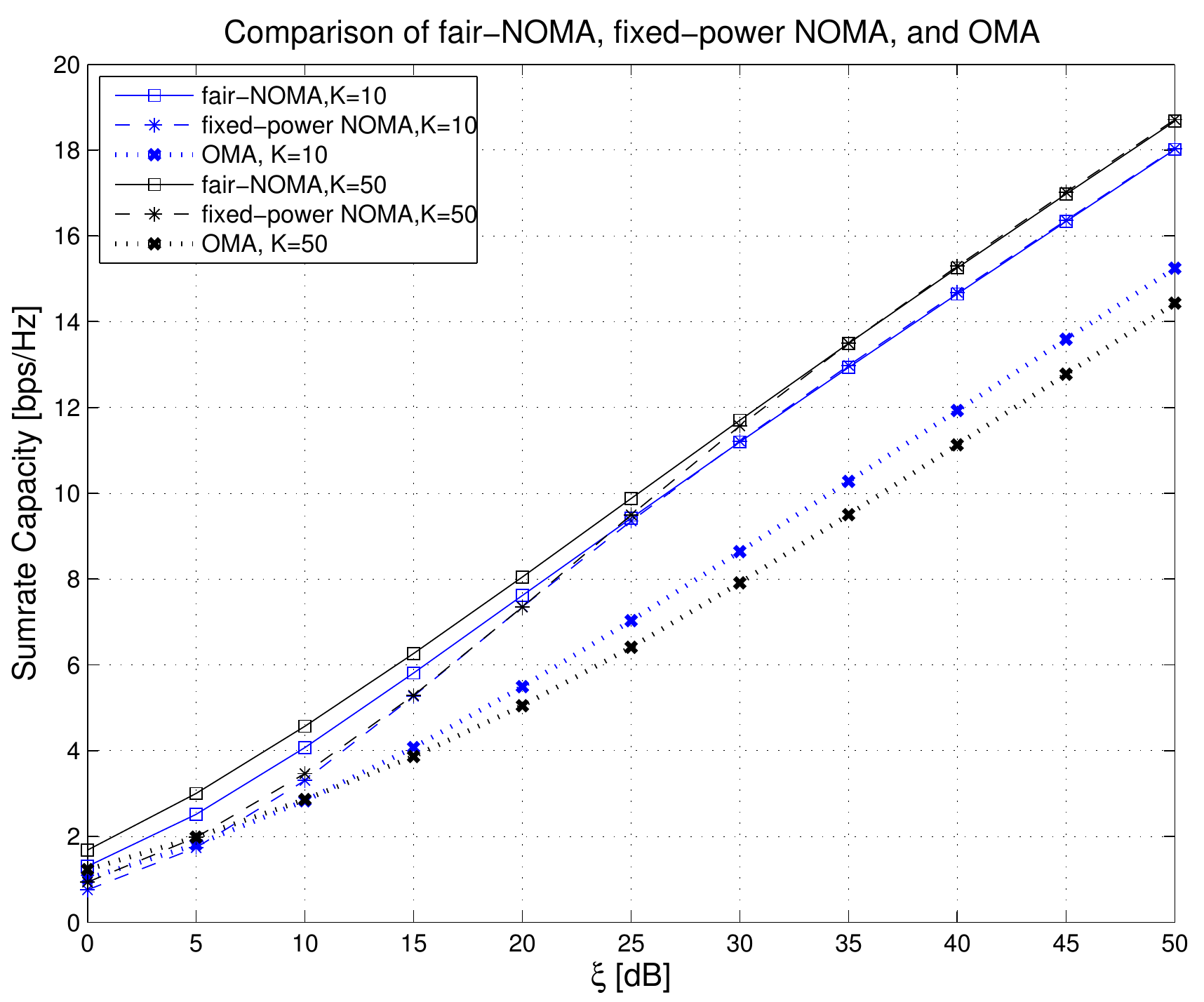}
\vspace{-3mm}
	\caption{\label{fig:NOMAmud_sumcompare} Comparison of Fair-NOMA, fixed-power NOMA, and OMA}
\end{figure}

Equation (\ref{eq:NOMAMUDDiff}) shows that the capacity gain made by pairing the nearest and furthest cell-edge users is slow in $K$, and is due to the combined gain in capacity achieved by the strongest user and loss in capacity by the weakest user. This makes sense from multiple points of view. The expected value of power allocation coefficient  $a_\text{sup}\rightarrow\frac{1}{2}$ when $K$ is large, due to the selection of the user with the weakest channel gain. In other words, as $K$ increases, the weakest user needs less power in NOMA to achieve the same capacity as it can using OMA. Hence, more power goes to the stronger user.  Figure \ref{fig:NOMAmud_deltaS} plots the simulation of $\mathbb{E}[\Delta S(a_\text{sup})]$ for $\xi=50$ dB, and the approximation given by (\ref{eq:NOMAMUDDiff}). Notice that the simulation and approximation seem to slightly diverge as the number of users increases. This is because the approximation in (\ref{eq:NOMAMUDDiff}) needs a sufficiently large value of $\xi$ as the number of users increases for the simulation and approximation to become tighter. However, $\xi=50$ dB was used because it is a large but still realistic value of $\xi$. Since the number of users $K$ cannot become arbitrarily large, the approximation remains tight for realistic values of $\xi$ and $K$. 
\begin{figure}[http]
	\centering
	\includegraphics[width=\columnwidth]{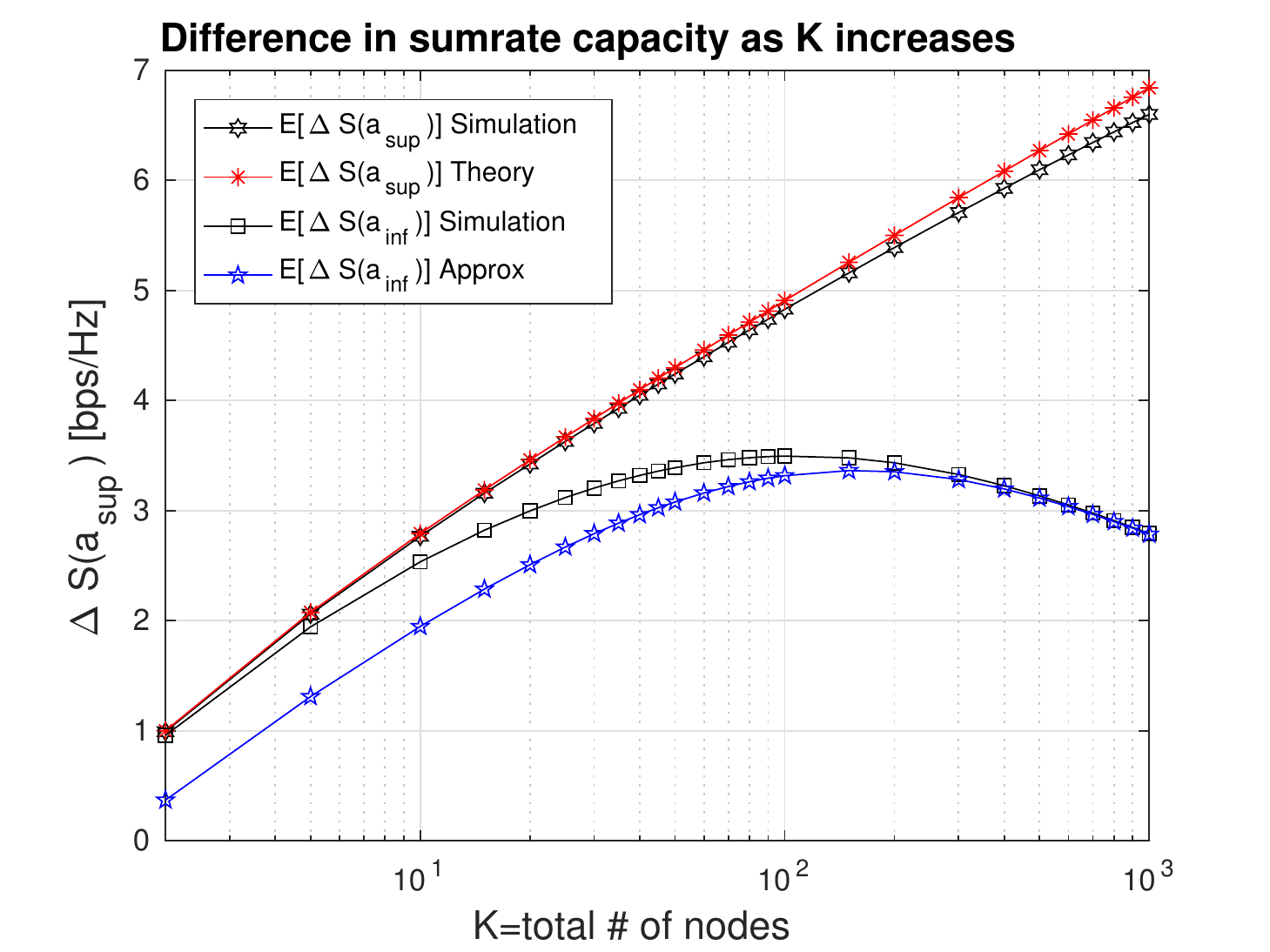}
\vspace{-3mm}

	\caption{\label{fig:NOMAmud_deltaS} Difference in ergodic capacity with opportunistic user-pairing; $\xi=50$ dB}
\end{figure}


\section{Multi-user NOMA in SISO Systems}\label{sec:MU-NOMA}
So far, the treatment of Fair-NOMA has focused on the two-user case.
Consider an OMA system, where $K$ users have their information transmitted over $K$ orthogonal time slots (or frequency bands) during a total time of $T$ (and bandwidth $B$). For each user $k$, the capacity of user $k$ is given by
\begin{equation}
	C_k^\text{O} =  \frac{1}{K}\log_2(1 + \xi|h_k|^2), \forall k=1,\ldots, K.
\end{equation}
When applying NOMA to this system, the information of each user occupies the entire time $T$ (bandwidth $B$) simultaneously. Hence, a superposition coding strategy must be used, in which all $K$ users must share the total transmit power $\xi$. User $k$ must perform SIC of each message that is intended for the other users $l$ that have weaker channel conditions than user $k$. The channel gains are ordered as $|h_1|^2<|h_2|^2<\cdots<|h_K|^2$. Lets define the power allocation coefficients $\{b_1, \ldots, b_K\}$, where $b_k$ is the power allocation coefficient for user $k$ and 
\begin{equation}\label{eq:munomapower}
	\sum_{k=1}^K b_k \leq 1.
\end{equation}
Therefore, the capacity of user $k$ for $1 \leq k \leq K$ is given by 
\begin{equation}
	C_k^\text{N}(b_1, \ldots, b_K) = \log_2\left(1 + \frac{b_k\xi|h_k|^2}{1 + \xi|h_k|^2\sum_{l=k+1}^K b_l } \right).
\end{equation}
 In order for $C_k^\text{N}(b_1, \ldots, b_K) > C_k^\text{O}$, the inequality must be solved for $b_k$ assuming that equation (\ref{eq:munomapower}) is true. Since user $K$ does not receive any interference power after decoding all of the other users' messages, solving for $b_K$ is straight forward.
\begin{align*}
	C_K^\text{O} \leq C_K^\text{N}(b_1, \ldots, b_K)
	\Rightarrow b_K \geq \frac{(1+\xi|h_K|^2)^\frac{1}{K}-1}{\xi|h_K|^2}.
\end{align*}
For users $k=1,\ldots, K-1$, the power allocation for each user is conditioned on $C_k^\text{O} \leq C_k^\text{N}(b_1, \ldots, b_K)$ which results in
\begin{align*}
	 b_k \geq \dfrac{[(1+\xi|h_k|^2)^\frac{1}{K}-1]\left(1+ \xi|h_k|^2\sum_{l=k+1}^K b_l \right)}{\xi|h_k|^2}.
\end{align*}
As expected, the power allocation of the users with weaker channel gains depend on the power allocation of the users with stronger channel gains. 

Notice that in the above derivation, the total power allocation was not necessarily used. 
Consider the case where $\sum_{k=1}^K a_k = 1$ and the case where user 1 capacity in OMA and NOMA are equal. Therefore, $C_1^\text{O} = C_1^\text{N}\Rightarrow$
\begin{align}
	&\Rightarrow\log_2(1+\xi|h_1|^2)^\frac{1}{K} = \log_2\left(1+\frac{a_1\xi|h_1|^2}{1 + \xi|h_1|^2\sum_{l=2}^K a_l }\right) \nonumber \\
	&\Rightarrow (1+\xi|h_1|^2)^\frac{1}{K} = \frac{1+\xi|h_1|^2}{1 + \xi|h_1|^2(1-a_1) }. 
\label{eq:powerint1}
\end{align}
Solving for $a_1$ gives
\begin{equation}\label{eq:powera1}
	a_1 = \frac{1+\xi|h_1|^2 - (1+\xi|h_1|^2)^\frac{K-1}{K}}{\xi|h_1|^2}.
\end{equation}
Note that both sides of equation (\ref{eq:powerint1}) are greater than 1 which means  $0<a_1<1, \forall \xi>0$.
Define $A_1 = 1-a_1$ as the sum of the interference coefficients to user 1. Therefore, 
\begin{equation}
	A_1 = \frac{(1+\xi|h_1|^2)^\frac{K-1}{K} -1}{\xi|h_1|^2},
\end{equation}
and $0<A_1<1$. In general, the power allocation coefficient required for the NOMA capacity of user $k$ to equal the OMA capacity of user $k$ can be derived by solving the equation
\begin{align}
	&C_k^\text{O} = C_k^\text{N}(a_1, \ldots, a_K)\nonumber\\
	\label{eq:powerboundK}\Rightarrow& (1+\xi|h_k|^2)^\frac{1}{K} = \frac{1 + A_{k-1}\xi|h_k|^2}{1 + (A_{k-1}-a_k)\xi|h_k|^2}, 
\end{align}
$\forall k\in\{2,\ldots,K\}, \xi>0$, where $A_{k-1} = 1-\sum_{l=1}^{k-1}a_l$. The following property for the set of power allocation coefficients $\{a_1,\ldots,a_K\}$ arises from solving equation (\ref{eq:powerboundK}).

\begin{property}\label{prop:multipower}
	If the set of power allocation coefficients $\{a_1,\ldots,a_K\}$ are derived from equations (\ref{eq:powerint1}) and (\ref{eq:powerboundK}), then 
	\begin{equation}
		a_k\in(0, 1), \hspace{5mm}\text{ and }\hspace{5mm}\sum_{k=1}^K a_k \leq 1.
	\end{equation}	 
\end{property}
\begin{proof}
	See appendix \ref{proof:multipower}.
\end{proof}
This is an important property, because it sets the precedent for the existence of a set of NOMA power allocation coefficients that (i) achieves at least OMA capacity for every user in the current transmission time period, and (ii) allows for at least one user to have a capacity greater than OMA capacity.

%

The power allocation coefficient $a_k$ considers interference received from users with higher channel gains to be at a maximum. However, the coefficients $b_k$ consider the minimum power allocation. 
Note that in power allocation for multiuser Fair-NOMA using $a_k$ coefficients, the allocation process begins with user having weakest channel (first user) and allocate enough power to have a capacity of at least equal to OMA capacity for the first user. Then, the process continues with the next user until all the power is allocated amongst all users, i.e., the last user with strongest channel receives the remaining power allocation that results in higher capacity than OMA for that user. When $b_k$ coefficients are used for power allocations, the power allocation process begins with the user with strongest channel, $K^{th}$ user and assign enough power to achieve the same capacity as OMA for that user. The process then continues with the next user until the process reaches the first user. 
%
%
%
%
%
%
Therefore, it is clear that 
\begin{align} 
	&b_k < a_k\\
	\text{and } & C_k^\text{N}(b_1,\ldots, b_K) < C_k^\text{N}(a_1,\ldots, a_K), \forall k.
\end{align} 
Hence, property \ref{prop:multipower} highlights that there always exists a power allocation scheme in the general multiuser NOMA case that always achieves higher capacity than OMA, while keeping the total transmit power to $\xi$. This minimum power allocation requirement is demonstrated in figure \ref{fig:totalpower}. The most interesting aspect of this result is that the same capacity of OMA can be achieved using Fair-NOMA with potentially much less total transmit power by using $b_k$ coefficients. This can be useful if the purpose of NOMA is to minimize the total transmit power in the network. 
\begin{figure}
\vspace{-5mm}
	\centering
	\includegraphics[width=\columnwidth]{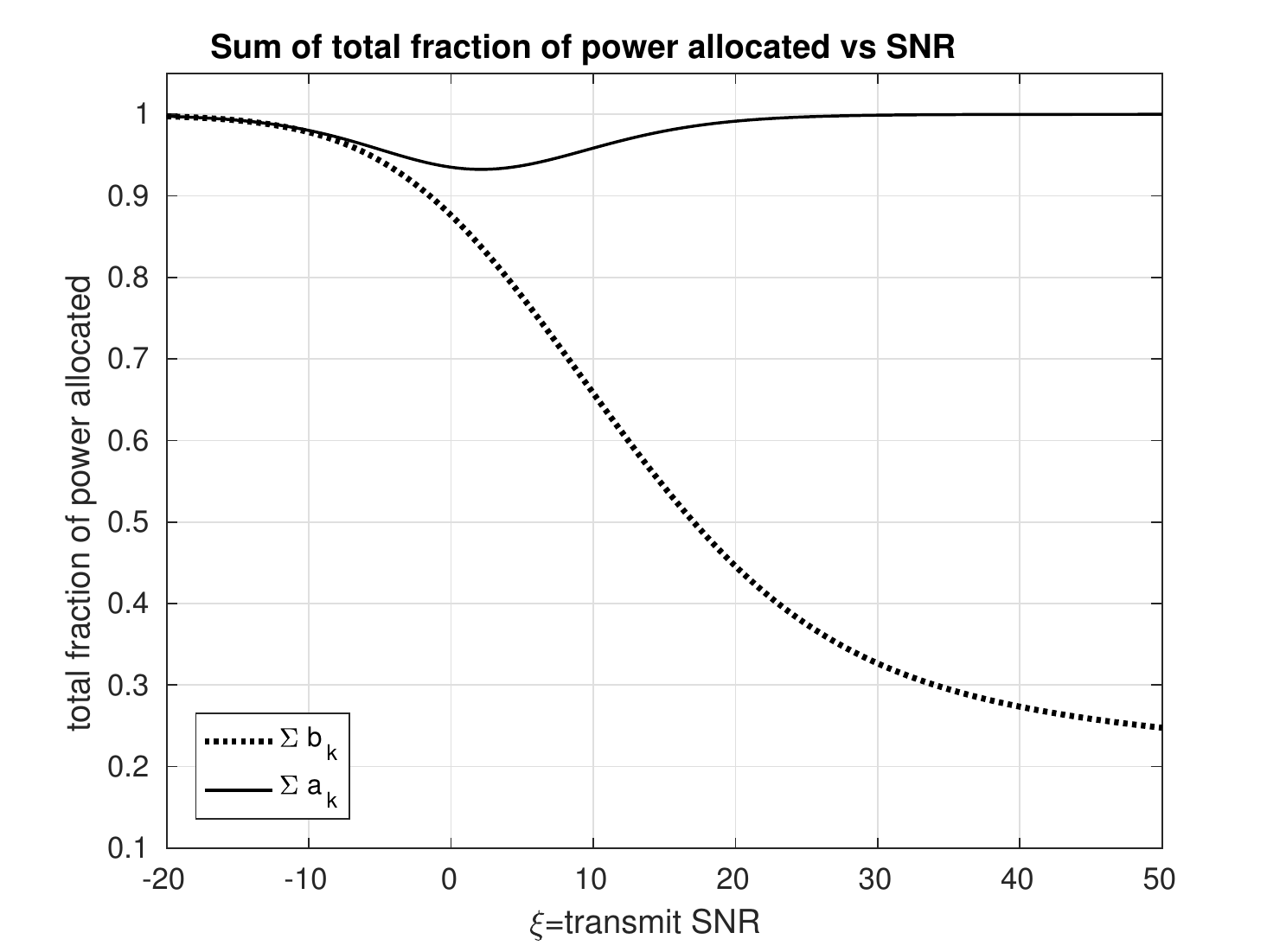}
\vspace{-3mm}
	\caption{\label{fig:totalpower} Minimum total power allocation in NOMA required to achieve capacity equal to OMA per user, $K=5$}
\vspace{0mm}
\end{figure}


\section{Conclusion and Future Work}\label{sec:conclusion}
Fair-NOMA approach is introduced which allows two paired users to achieve capacity greater than or equal to the capacity with OMA. Given the power allocation set $\mathcal{A}_\text{FN}$  for this scheme, the ergodic capacity for the infimum and supremum of this set is derived for each user, and the expected asymptotic capacity gain is found to be 1 bps/Hz. The outage probability was also derived  and it is shown that when $a=a_\text{inf}$, the outage performance of the weaker user significantly improves over OMA, where as the outage performance of the stronger user improves by at most roughly 2dB. 

Fair-NOMA is applied to opportunistic user-pairing  and the exact  capacity gain  is computed. The performance of Fair-NOMA is compared with a fixed-power NOMA approach to show that even when the power allocation coefficient $a=a_\text{sup}$ becomes less than the fixed-power allocation coefficient, the capacity gain is the same at high SNR, while Fair-NOMA clearly outperforms the fixed-power approach at low SNR.

The concept of  Fair-NOMA can be extended to  MIMO systems.  In \cite{MIMONOMA:DAP}, a similar result is found for the approximate expected capacity gap of a 2-user 2-by-2 MIMO NOMA system. In order to eliminate the existing possibility that NOMA does not outperform OMA in capacity for any user, the Fair-NOMA approach can be applied to users that are utilizing the same degree of freedom from the base-station. By ordering the composite channel gains, which include the transmit and receive beamforming applied to the channel, $K$ users on the same transmit beam can have their signals superpositioned, and then SIC can be done at their receivers to obtain their own signal with minimum interference. Receive beamforming is used to eliminate the interference from the transmit beams' signals. The power allocation region can then be derived in the same manner as in Section \ref{sec:MU-NOMA}, and NOMA can then be used to either increase the capacity gap as is done in \cite{MIMONOMA:DAP}, or to minimize the transmit power required to achieve the same capacity as in OMA, similar to what was done in Section \ref{sec:MU-NOMA}.

Finally, it is necessary to demonstrate a full system analysis and simulation of the impact of NOMA on bit-error rate (BER). It has been shown that the BER is very tightly approximated by the outage probability in \cite{MUDDIV:ZT}, and therefore a tight approximation of the BER performance is given in this work. However, the transmission of signals using superposition coding and different information rates for the users are factors that impact the analysis of BER for the signal of the weaker user, especially because constellation sizes will most likely be different for the two superpositioned signals.


\appendices
\renewcommand{\theequation}{\thesection.\arabic{equation}}
\section{Proof of Theorem \ref{thm:thm1}}\label{proof:thm1}
\begin{IEEEproof}
	For the case when $a = a_\text{inf}$, the proof is trivial. Proving for the case when $a=a_\text{sup}$ then suffices to show it is true for all $a\in\mathcal{A}_\text{FN}$, because the $C_1^\text{N}(a)$ performance is lower-bounded by the case when $a=a_\text{sup}$, while for $C_2^\text{N}(a)$ the performance will only improve for $a>a_\text{inf}$. 
	In order for $S$ to be monotonically increasing function of $|h_i|^2$, it must be shown that $dS/d|h_i|^2 > 0$, $\forall |h_i|^2$. 
	In the case of $|h_2|^2$, $C_1^\text{N}(a_\text{sup})$ does not factor in, so
		$\frac{dS}{d|h_2|^2} = \frac{dC_2^\text{N}(a_\text{sup}) }{d|h_2|^2} = \frac{a_\text{sup}\xi}{a_\text{sup}\xi|h_2|^2} > 0,$
	$\forall |h_2|^2$. 
	The case of $|h_1|^2$ goes as follows.
	\begin{align*}
		 &\frac{dS}{d|h_1|^2} 
		 =\frac{1}{\ln 2}\left[\frac{\xi}{2(1+\xi|h_1|^2)}\right.\\
		  & + \left. \frac{\frac{\xi|h_2|^2}{2|h_1|^2\sqrt{1+\xi|h_1|^2}} + \frac{|h_2|^2}{|h_1|^4}(1-\sqrt{1+\xi|h_1|^2}  )}{1+\frac{|h_2|^2}{|h_1|^2}(\sqrt{1+\xi|h_1|^2} - 1)}\right]\\
		 =& \frac{\left\{\begin{array}{l}\xi(|h_1|^4+|h_1|^2|h_2|^2[(1+\xi|h_1|^2)^{\frac{1}{2}}-1]\\ + 2|h_2|^2(1+\xi|h_1|^2) + \xi|h_1|^2|h_2|^2(1+\xi|h_1|^2)^{\frac{1}{2}}\\ - 2|h_2|^2(1+\xi|h_1|^2)^{\frac{3}{2}}\end{array} \right\}  }{ 2|h_1|^2\ln(2)(1+\xi|h_1|^2)(|h_1|^2+|h_2|^2(\sqrt{1+\xi|h_1|^2}-1) )  }
	\end{align*}
	The numerator above can be simplified as
	\begin{align*}
		& \xi|h_1|^4 + |h_2|^2  \\
		 &\times[ 2\xi|h_1|^2\sqrt{1+\xi|h_1|^2} + \xi|h_1|^2 + 2 - 2(1+\xi|h_1|^2)^\frac{3}{2}].
	\end{align*}
	The value inside the square brackets can be simplified to
	\begin{align*}
		=& 2+\xi|h_1|^2 - 2\sqrt{1+\xi|h_1|^2}.
	\end{align*}
	Since $2+\xi|h_1|^2 - 2\sqrt{1+\xi|h_1|^2} \geq 0$ because $\xi^2|h_1|^4 \geq 0$, then 
	\begin{align*}
		\Rightarrow & \xi|h_1|^4+|h_2|^2(2+\xi|h_1|^2 - 2\sqrt{1+\xi|h_1|^2}) > 0\\
		\Rightarrow&\frac{dS}{d|h_1|^2}> 0.
	\end{align*}
	Since $|h_1|^2<|h_2|^2$, $\frac{dS}{d|h_1|^2}>0$, and $\frac{dS}{d|h_2|^2}>0$, then $S$ is a monotonically increasing function with respect to $|h_1|^2$ and $|h_2|^2$.
\end{IEEEproof}

\section{Proof of Theorem \ref{thm:DiffCap1}}\label{proof:DiffCap1}
	\begin{IEEEproof}
		For $\xi\gg 1$, 
		\begin{align*}
			\Delta C_1(a_\text{inf})\approx \Delta C_2(a_\text{sup}) \approx \tfrac{1}{2}\log_2(|h_2|^2)-\tfrac{1}{2}\log_2(|h_1|^2).
		\end{align*}The expected value of $\Delta C_1(a_\text{inf})$ and $\Delta C_2(a_\text{sup})$ is then
		\begin{align*}
			&\mathbb{E}\left[\frac{1}{2}\log_2(\frac{|h_2|^2}{|h_1|^2})\right] \approx \int_0^\infty \hspace{-2mm}\int_{0}^{x_2}\frac{1}{\beta^2}e^{-\frac{x_1+x_2}{\beta}}\log_2(x_2)dx_1dx_2 \\
			 &- \int_0^\infty\hspace{-2mm}\int_{x_1}^\infty \frac{1}{\beta^2}e^{-\frac{x_1+x_2}{\beta}}\log_2(x_1)dx_2dx_1\\
			=& \int_0^\infty\frac{1}{\beta}e^{-\frac{x_2}{\beta}}\log_2(x_2)dx_2 - \int_0^\infty\frac{1}{\beta}e^{-\frac{2x_2}{\beta}}\log_2(x_2)dx_2 \\
			 &- \int_0^\infty\frac{1}{\beta}e^{-\frac{2x_1}{\beta}}\log_2(x_1)dx_1\\
			\stackrel{(a)}{=}& \int_0^\infty\frac{1}{\beta}e^{-\frac{x_2}{\beta}}\log_2(x_2)dx_2 
			  - 2\int_0^\infty\frac{1}{\beta}e^{-\frac{2x_2}{\beta}}\log_2(x_2)dx_2 \\
			\stackrel{(b)}{=}& \int_0^\infty\frac{1}{\beta}e^{-\frac{x_2}{\beta}}\log_2(x_2)dx_2 - \int_0^\infty\frac{1}{\beta}e^{-\frac{x}{\beta}}\log_2(x)dx\\
			 & + \int_0^\infty \frac{1}{\beta}e^{-\frac{x}{\beta}}\log_2(2)dx
			 \stackrel{(c)}{=} 1,
			\end{align*}
			where $(a)$ is true because the second two integrals are actually the same integral adding together, $(b)$ is true by making the substitution $x= 2x_2$. and $(c)$ is true because the first two integrals are the same integral subtracting each other. Since $S_\text{N}(a)$ is a monotonically increasing function of $a$, and $S_\text{N}(a_\text{inf})-S_\text{O} = \Delta C_1(a_\text{inf})$ and $S_\text{N}(a_\text{sup})-S_\text{O}=\Delta C_2(a_\text{sup})$, then when $\xi\gg 1$, $\Delta S = S_\text{N}(a)-S_\text{O}\approx 1, \forall a\in\mathcal{A}_\text{FN}$.

	\end{IEEEproof}

\section{Proof of Outage Probability Results} \label{proof:Outage_Probabilities}
\subsection{Proof of property \ref{prop:outage}(\ref{item:outage1})}
\begin{IEEEproof} 
	$p_{\text{N},1}^\text{out}(a_\text{inf}) $
	\begin{align}
		= & \mathrm{Pr}\left\{ \log_2\left(\tfrac{1+\xi|h_1|^2}{1+a_\text{inf}\xi|h_1|^2}\right) < R_0\right\}\\
		\label{eq:p1ainf1} = & \mathrm{Pr}\left\{ |h_1|^2 < \tfrac{|h_2|^2(2^{R_0}-1)}{\xi|h_2|^2 + 2^{R_0}(1-\sqrt{1+\xi|h_2|^2})} \right\}
	\end{align}
	Since $|h_1|^2<|h_2|^2$, then there are two cases as 
	\begin{align}
		&|h_2|^2 \lessgtr \tfrac{|h_2|^2(2^{R_0}-1)}{\xi|h_2|^2 + 2^{R_0}(1-\sqrt{1+\xi|h_2|^2})} = \alpha_1.
	\end{align}
	Solving above for $|h_2|^2$ gives
	\begin{align}
		\Longrightarrow& |h_2|^2 \lessgtr \tfrac{4^{R_0}-2}{2\xi}+\sqrt{\tfrac{4^{R_0}-1}{\xi^2} + \tfrac{(4^{R_0}-2)^2}{4\xi^2}}=\alpha_2.
	\end{align}
	This allows the event in equation (\ref{eq:p1ainf1}) to be written as the two mutually exclusive events given in 
	\begin{align}
		&\left\{ |h_1|^2 < \alpha_1\right\} \\
		&=\left\{ |h_1|^2<|h_2|^2, |h_2|^2< \alpha_2 \right\} \bigcup \left\{ |h_1|^2<\alpha_1, |h_2|^2 > \alpha_2 \right\}. \nonumber
	\end{align}
	The probability in equation (\ref{eq:p1ainf1}) can then be written as $\mathrm{Pr}\{|h_1|^2<\alpha_1\}=\mathrm{Pr}\{  |h_1|^2<|h_2|^2, |h_2|^2< \alpha_2 \} + \mathrm{Pr}\{|h_1|^2<\alpha_1, |h_2|^2 > \alpha_2\}$. The first probability is equal to
	\begin{align}
		&\mathrm{Pr}\{  |h_1|^2<|h_2|^2, |h_2|^2< \alpha_2 \}=\int_0^{\alpha_2}\int_0^{x_2}\frac{2}{\beta}e^{-\frac{x_1+x_2}{\beta}}dx_1dx_2 \nonumber\\
		\label{eq:p1ainf2}&=1+e^{-\frac{2\alpha_2}{\beta}}-2e^{-\frac{\alpha_2}{\beta}}.
	\end{align}
	The second probability is found to be
	\begin{align}
		& \mathrm{Pr}\{|h_1|^2<\alpha_1, |h_2|^2 > \alpha_2\}=\int_{\alpha_2}^\infty\int_0^{\alpha_1}\frac{2}{\beta}e^{-\frac{x_1+x_2}{\beta}}dx_1dx_2 \nonumber\\
		\label{eq:p1ainf3}&= 2e^{-\frac{\alpha_2}{\beta}} -\frac{2}{\beta}\int_{\alpha_2}^\infty e^{-\frac{x_2+\alpha_1}{\beta}}dx_2,
	\end{align}
	where the integral in equation (\ref{eq:p1ainf3}) has no known closed-form solution. Combining equations (\ref{eq:p1ainf2}) and (\ref{eq:p1ainf3}), gives
	\begin{align}
		p_{\text{N},1}^\text{out}(a_\text{inf})=1+e^{-\frac{2\alpha_2}{\beta}}-\frac{2}{\beta}\int_{\alpha_2}^\infty e^{-\frac{x_2+\alpha_1}{\beta}}dx_2
	\end{align}	
\end{IEEEproof}
\subsection{Proof of property \ref{prop:outage}(\ref{item:outage2})}
\begin{IEEEproof}
	$p_{\text{N},2}^\text{out}(a_\text{sup})$
	\begin{align}
		=& \mathrm{Pr}\left\{  \log_2(1+\tfrac{|h_2|^2}{|h_1|^2}(\sqrt{1+\xi|h_1|^2} -1)) < R_0 \right\}\\
			\label{eq:pout2noma1}=& \mathrm{Pr}\left\{ |h_1|^2 > \tfrac{\xi|h_2|^4}{(2^{R_0}-1)^2} - \tfrac{2|h_2|^2}{2^{R_0}-1}  \right\}.
	\end{align}
	Since $0<|h_1|^2<|h_2|^2$ is always true, then the domain of $|h_2|^2$ that makes the statement $\frac{\xi|h_2|^4}{(2^{R_0}-1)^2} - \frac{2|h_2|^2}{2^{R_0}-1}>0$ true or false must be found, and thus gives us two intervals for $|h_2|^2$. 
	\begin{align}
		&\tfrac{\xi|h_2|^4}{(2^{R_0}-1)^2} - \tfrac{2|h_2|^2}{2^{R_0}-1} \lessgtr 0\\
		\Longrightarrow &|h_2|^2 \lessgtr \tfrac{2(2^{R_0}-1)}{\xi}. 
	\end{align}
	For the case of $|h_2|^2 < \frac{2(2^{R_0}-1)}{\xi}$, which gives $\frac{\xi|h_2|^4}{(2^{R_0}-1)^2} - \frac{2|h_2|^2}{2^{R_0}-1} < 0$, the event is explicitly written as
	\begin{align}
		\label{eq:outage1asup}\mathcal{A}_1^\text{out}=\left\{0<|h_1|^2<|h_2|^2, 0<|h_2|^2<\tfrac{2(2^{R_0}-1)}{\xi}\right\}.
	\end{align}
	For the case of $|h_2|^2 > \frac{2(2^{R_0}-1)}{\xi}$, 
	the interval for $|h_1|^2$ is $\frac{\xi|h_2|^4}{(2^{R_0}-1)^2} - \frac{2|h_2|^2}{2^{R_0}-1}< |h_1|^2<|h_2|^2$, so it must also be true that $|h_2|^2 > \frac{\xi|h_2|^4}{(2^{R_0}-1)^2} - \frac{2|h_2|^2}{2^{R_0}-1}$. This gives $|h_2|^2 < \frac{(2^{R_0}-1)^2 + 2(2^{R_0}-1)}{\xi} = \frac{4^{R_0}-1}{\xi}$, and therefore the interval for this event is explicitly written as
	\begin{align}
		\label{eq:outage2asup} \mathcal{A}_2^\text{out}=\left\{ \phantom{\frac{A}{B}}\right. \hspace{-2mm}& \tfrac{\xi|h_2|^4-2|h_2|^2(2^{R_0}-1)}{(2^{R_0}-1)^2}< |h_1|^2<|h_2|^2,  &\\
		 &\left.\tfrac{2(2^{R_0}-1)}{\xi}<|h_2|^2< \frac{4^{R_0}-1}{\xi} \right\}&\nonumber
	\end{align}
	Now the probability above can be derived by computing the probabilities of the two disjoint regions as 
	\begin{align}
		&\mathrm{Pr}\left\{ |h_1|^2 > \tfrac{\xi|h_2|^4}{(2^{R_0}-1)^2} - \tfrac{2|h_2|^2}{2^{R_0}-1}  \right\}= \mathrm{Pr}\{\mathcal{A}_1^\text{out}\} + \mathrm{Pr}\{\mathcal{A}_2^\text{out}\}.
	\end{align}
	The first probability is computed by
	\begin{align}
		&\mathrm{Pr}\{\mathcal{A}_1^\text{out}\} {\color{black} = \int_0^{\frac{2(2^{R_0}-1)}{\xi}}\int_0^{x_2}\frac{2}{\beta^2}e^{-\frac{x_1+x_2}{\beta}}dx_1dx_2\nonumber}\\
		\label{eq:p2asup1}&=1+e^{-\frac{4(2^{R_0}-1)}{\beta\xi}} - 2e^{-\frac{2(2^{R_0}-1)}{\beta\xi}}.
	\end{align}
	%
	Let $K_1=\frac{2(2^{R_0}-1)}{\xi}$ and $K_2=\frac{4^{R_0}-1}{\xi}$. Then the second probability is given by
	\begin{align}
		&\mathrm{Pr}\{\mathcal{A}_2^\text{out}\} = \int_{K_1}^{K_2}\int_{\frac{\xi x_2^2-2x_2^{\phantom{2}}(2^{R_0}-1)}{(2^{R_0}-1)^2}}^{x_2}\frac{2}{\beta^2}e^{-\frac{x_1+x_2}{\beta}}dx_1dx_2\\
		\label{eq:p2asup2p1}&= \int_{K_1}^{K_2} \frac{2}{\beta}\left( e^{-\left(\frac{x_2}{\beta}+\frac{\xi x_2^2-2x_2^{\phantom{2}}(2^{R_0}-1)}{\beta(2^{R_0}-1)^2}\right) } - e^{-\frac{2x_2}{\beta}} \right)dx_2
	\end{align}
	The second term in the integral in equation (\ref{eq:p2asup2p1}) can be easily computed to be 
	\begin{align}
		&\label{eq:p2asup2p1-1}\int_{K_1}^{K_2} \frac{2}{\beta}e^{-\frac{2x_2}{\beta}}dx_2=-e^{- \frac{2(4^{R_0}-1)}{\beta\xi}} + e^{-\frac{4(2^{R_0}-1)}{\beta\xi}}.
	\end{align}
	The first integral in equation (\ref{eq:p2asup2p1}) is computed by completing the square in the exponent as 
	\begin{align}
		& \int_{K_1}^{K_2} \frac{2}{\beta} e^{-\left(\frac{x_2}{\beta}+\frac{\xi x_2^2-2x_2^{\phantom{2}}(2^{R_0}-1)}{\beta(2^{R_0}-1)^2}\right) } dx_2\\
		\label{eq:p2asup2p2}&= {\color{black}\int_{K_1}^{K_2} \frac{2}{\beta} e^{-\frac{\xi}{\beta(2^{R_0}-1)^2}\left(x_2^2 + \frac{x_2^{\phantom{2}}(2^{R_0}-1)(2^{R_0}-3)}{\xi}\right) } dx_2.}
	\end{align}
	Since 
	\begin{align}
		&x_2^2 + \tfrac{x_2(2^{R_0}-1)(2^{R_0}-3)}{\xi}\\
		=&\left(x_2 + \tfrac{(2^{R_0}-1)(2^{R_0}-3)}{2\xi}\right)^2 -  \left(\tfrac{(2^{R_0}-1)(2^{R_0}-3)}{2\xi}\right)^2\hspace{-1mm},
	\end{align}
	then equation (\ref{eq:p2asup2p2}) equals
	\begin{align}
		\label{eq:p2asup2p3}&=  \frac{2}{\beta} e^{\frac{(2^{R_0}-3)^2}{4\beta\xi}}\int_{K_1}^{K_2} e^{-\frac{\xi}{\beta(2^{R_0}-1)^2}\left(x_2^{\phantom{2}} + \frac{(2^{R_0}-1)(2^{R_0}-3)}{2\xi}\right)^2} dx_2	.
	\end{align}
	By using the substitution 
	\begin{align}
		\label{eq:p2outsub}u(x_2) = \tfrac{1}{2^{R_0}-1}\sqrt{\tfrac{\xi}{\beta}}\left(x_2 + \tfrac{(2^{R_0}-1)(2^{R_0}-3)}{2\xi}\right),
	\end{align}
	the integral in equation (\ref{eq:p2asup2p3}) equals
	\begin{align}
		\label{eq:p2asup2p4}&=\frac{2}{\beta} e^{\frac{(2^{R_0}-3)^2}{4\beta\xi}}\int_{u(K_1)}^{u(K_2)} e^{-u^2}\cdot(2^{R_0}-1)\sqrt{\frac{\beta}{\xi}}du.\\
		&=\label{eq:p2asup2p5}(2^{R_0}-1)e^{\frac{(2^{R_0}-3)^2}{4\beta\xi}}\sqrt{\frac{\pi}{\beta\xi}}\cdot[\mathrm{erfc}\left(u(K_1)\right) - \mathrm{erfc}\left(u(K_2)\right)], 
	\end{align}
	where $u(x)$ is obtained by equation (\ref{eq:p2outsub}), and thus $u(K_1) = \frac{2^{R_0}+1}{2\sqrt{\beta\xi}}$, $u(K_2) = \frac{3(2^{R_0})-1}{2\sqrt{\beta\xi}}$, and $\mathrm{erfc}(z)=\frac{2}{\sqrt{\pi}}\int_z^\infty e^{-u^2} du$ is the complementary error function.
	Thus, combining equations (\ref{eq:p2asup1}, \ref{eq:p2asup2p1-1}, \ref{eq:p2asup2p5}) results in equation (\ref{eq:outage2}).
\end{IEEEproof}

\section{Proof of Theorem \ref{thm:NOMAminmax}}\label{proof:NOMAminmax}
\begin{IEEEproof}
When $\xi |h_0|^2 \gg 1$, $\Delta S(a_\text{sup}) \approx \frac{1}{2}(\log_2(|h_M|^2)-\log_2(|h_0|^2))$. Therefore by equation (\ref{eq:expminmaxlong}),

\begin{figure*}[http!]
	\begin{align}
	\mathbb{E}\left[ \Delta S(a_\text{sup}) \right] \approx  &\int_0^\infty\int_0^{x_M}\tfrac{1}{2}\log_2 (x_M) \tfrac{K(K-1)}{\beta^2}e^{-\frac{x_0+x_M}{\beta}}(e^{-\frac{x_0}{\beta}}-e^{-\frac{x_M}{\beta}})^{K-2}dx_0dx_M \nonumber\\
	  \label{eq:expminmaxlong}&- \int_0^\infty\int_{x_0}^\infty\tfrac{1}{2}\log_2 (x_0) \tfrac{K(K-1)}{\beta^2}e^{-\frac{x_0+x_M}{\beta}}(e^{-\frac{x_0}{\beta}}-e^{-\frac{x_M}{\beta}})^{K-2}dx_Mdx_0
	\end{align}
	\hrule
	\vspace{-5mm}
\end{figure*} \vspace{-5mm}

\begin{align*}
	&\mathbb{E}\left[ \Delta S(a_\text{sup}) \right] \\
	  \approx& \int_0^\infty\log_2 (x_M) \frac{K}{2\beta}e^{-\frac{x_M}{\beta}}(1-e^{-\frac{x_M}{\beta}})^{K-1}dx_M \\
	   &- \int_0^\infty\log_2 (x_0) \frac{K}{2\beta}e^{-\frac{Kx_0}{\beta}}dx_0 \\ 
	  =& \int_0^\infty\log_2 (x_M) \frac{K}{2\beta}e^{-\frac{x_M}{\beta}} \sum_{n=0}^{K-1} \tbinom{K-1}{n}(-1)^{n}e^{-\frac{nx_M}{\beta}}dx_M \\
	   &- \int_0^\infty\log_2 (x_0) \frac{K}{2\beta}e^{-\frac{Kx_0}{\beta}}dx_0 \\
	  =& \sum_{n=0}^{K-1} \int_0^\infty\log_2 \left(\frac{x}{n+1}\right) \frac{K}{2(n+1)\beta} \tbinom{K-1}{n}(-1)^{n}e^{-\frac{x}{\beta}}dx \\
	   &- \int_0^\infty\log_2\left(\frac{x}{K}\right) \frac{1}{2\beta}e^{-\frac{x}{\beta}}dx \\	  
	  =& \int_0^\infty \log_2 \left(\frac{K}{x}\cdot \prod_{n=0}^{K-1} \left(\frac{x}{n+1}\right)^{\binom{K}{n+1}(-1)^{n}} \right) \frac{1}{2\beta}e^{-\frac{x}{\beta}}dx \\
	  =&  \frac{1}{2}\log_2(K) + \frac{1}{2}\sum_{m=1}^K\tbinom{K}{m}(-1)^{m} \log_2 (m) .
\end{align*}
\end{IEEEproof}


\section{Proof of Property \ref{prop:multipower}}\label{proof:multipower}
\begin{IEEEproof}
It is already established that $a_1, A_1\in(0,1)$. The power allocation coefficient for user 2 is found by the equation
\begin{align}\label{eq:poweru2derive}
	(1+\xi|h_2|^2)^\frac{1}{K} =& \frac{1+A_1\xi|h_2|^2}{1+(A_1-a_2)\xi|h_2|^2}\nonumber\\
	\Longrightarrow a_2 =& \frac{(1+A_1\xi|h_2|^2)[(1+\xi|h_2|^2)^\frac{1}{K}-1]}{\xi|h_2|^2(1+\xi|h_2|^2)^\frac{1}{K}}.
\end{align}
If the following is true
\begin{equation}\label{ineq:poweru2}
		1< \frac{1+A_1\xi|h_2|^2}{1+(A_1-a_2)\xi|h_2|^2} < 1+A_1\xi|h_2|^2,
\end{equation}
then clearly $a_2\in(0, A_1)$. However, for equation (\ref{eq:poweru2derive}) and inequality (\ref{ineq:poweru2}) to be true, it must be true that
\begin{equation}
			1< (1+\xi|h_2|^2)^\frac{1}{K} < 1+A_1\xi|h_2|^2. 
\end{equation}
It is trivial to show that $1<(1+\xi|h_2|^2)^\frac{1}{K}, \forall \xi, |h_2|^2>0$. To show that $(1+\xi|h_2|^2)^\frac{1}{K}<1+A_1\xi|h_2|^2, \forall \xi>0$, the inequality is rearranged so that
\begin{equation}
	\gamma_2 < A_1,
\end{equation}
where
\begin{equation}
	\gamma_k = \frac{(1+\xi|h_k|^2)^\frac{1}{K}-1}{\xi|h_k|^2}.
\end{equation}
The inequality $\gamma_2 < A_1$ is clearly true because $\gamma_2 < \gamma_1$, and $\gamma_1 < A_1$ because $(1+\xi|h_1|^2)^\frac{m}{K} < (1+\xi|h_1|^2)^\frac{K-1}{K},\forall m<K-1$. Therefore, equation (\ref{eq:poweru2derive}) and inequality (\ref{ineq:poweru2}) are true. In a similar manner, in order for the power allocation coefficient $a_k$ for user $k$ to be less than total interference $A_{k-1}$ received by user $k-1$, the following must be true:
\begin{align}
	\label{eq:munomapowerk1}&a_k = \frac{(1+A_{k-1}\xi|h_k|^2)[(1+\xi|h_k|^2)^\frac{1}{K}-1]}{\xi|h_k|^2(1+\xi|h_k|^2)^\frac{1}{K}},\\
	\label{eq:munomapowerk2}&1< \frac{1+A_{k-1}\xi|h_k|^2}{1+(A_{k-1}-a_k)\xi|h_k|^2} < 1+A_{k-1}\xi|h_k|^2,\\
	\label{eq:munomapowerk3}&1< (1+\xi|h_k|^2)^\frac{1}{K} < 1+A_{k-1}\xi|h_k|^2.
\end{align}
Equation (\ref{eq:munomapowerk1}) is true by solving eq. (\ref{eq:powerboundK}), while (\ref{eq:munomapowerk2}) states that $a_k\in(0,A_{k-1})$ and (\ref{eq:munomapowerk3}) requires that user $k$'s OMA capacity is feasible within $a_k\in(0,A_{k-1})$, given the channel condition of user $k$. Therefore, (\ref{eq:munomapowerk3}) leads to 
\begin{align}
	&\gamma_k < A_{k-1} = A_{k-2}-a_{k-1} 
		= 	\frac{A_{k-2}-\gamma_{k-1}}{(1+\xi|h_{k-1}|^2)^\frac{1}{K}} \nonumber\\
\label{eq:lbfrac0} &\Longrightarrow A_{k-2}> \nonumber\\
&\frac{(1+\xi|h_k|^2)^\frac{1}{K}-1}{\xi|h_k|^2}(1+\xi|h_{k-1}|^2)^\frac{1}{K}  +\frac{(1+\xi|h_{k-1}|^2)^\frac{1}{K}-1}{\xi|h_{k-1}|^2}
\end{align}
Since the function
\begin{equation}\label{eq:monotonicfrac}
	f(x) = \frac{(1+x)^\frac{m}{K}-1}{x}, \forall m<K, m\text{ and }K\in\mathbb{N}
\end{equation}
is a monotonically decreasing function of $x$, then
{\small 
\begin{align}
	\label{eq:lbfrac1} &\frac{(1+\xi|h_k|^2)^\frac{1}{K}-1}{\xi|h_k|^2}(1+\xi|h_{k-1}|^2)^\frac{1}{K} 
			+ \frac{(1+\xi|h_{k-1}|^2)^\frac{1}{K}-1}{\xi|h_{k-1}|^2} \\
			&< \frac{(1+\xi|h_{k-1}|^2)^\frac{1}{K}-1}{\xi|h_{k-1}|^2}(1+\xi|h_{k-1}|^2)^\frac{1}{K} 
			+\frac{(1+\xi|h_{k-1}|^2)^\frac{1}{K}-1}{\xi|h_{k-1}|^2} \\
\label{eq:lbfrac2}					&= \frac{(1+\xi|h_{k-1}|^2)^\frac{2}{K}-1}{\xi|h_{k-1}|^2} < A_{k-2} 
			= A_{k-3}-a_{k-2} \nonumber\\
			&= \frac{ A_{k-3}-\gamma_{k-2}}{(1+\xi|h_{k-2}|^2)^\frac{1}{K}} \\
\label{eq:lbfrac3}	&	\Longrightarrow	A_{k-3}> \frac{(1+\xi|h_{k-1}|^2)^\frac{2}{K}-1}{\xi|h_{k-1}|^2}  (1+\xi|h_{k-2}|^2)^\frac{1}{K} \\
			& +\frac{(1+\xi|h_{k-2}|^2)^\frac{1}{K}-1}{\xi|h_{k-2}|^2}.
\end{align}}
The inequality in (\ref{eq:lbfrac3}) has the same form as the inequality in (\ref{eq:lbfrac0}), so the same steps taken in inequalities (\ref{eq:lbfrac1}) and (\ref{eq:lbfrac2}) can be used repeatedly, until the following is obtained
\begin{equation}
	\frac{(1+\xi|h_{1}|^2)^\frac{k-1}{K}-1}{\xi|h_{1}|^2} \leq A_{1},
\end{equation}
which is true $\forall k\leq K$. Hence, this series of inequalities shows that the transmit power allocation coefficient $a_k$ required for user $k$ to achieve OMA capacity is always less than the total interference coefficient received by user $k-1$, which equals the total fraction of power available for users $k,\ldots,K$. 

\end{IEEEproof}



\begin{IEEEbiography}[{\includegraphics[width=1in,height=1.25in,clip,keepaspectratio]{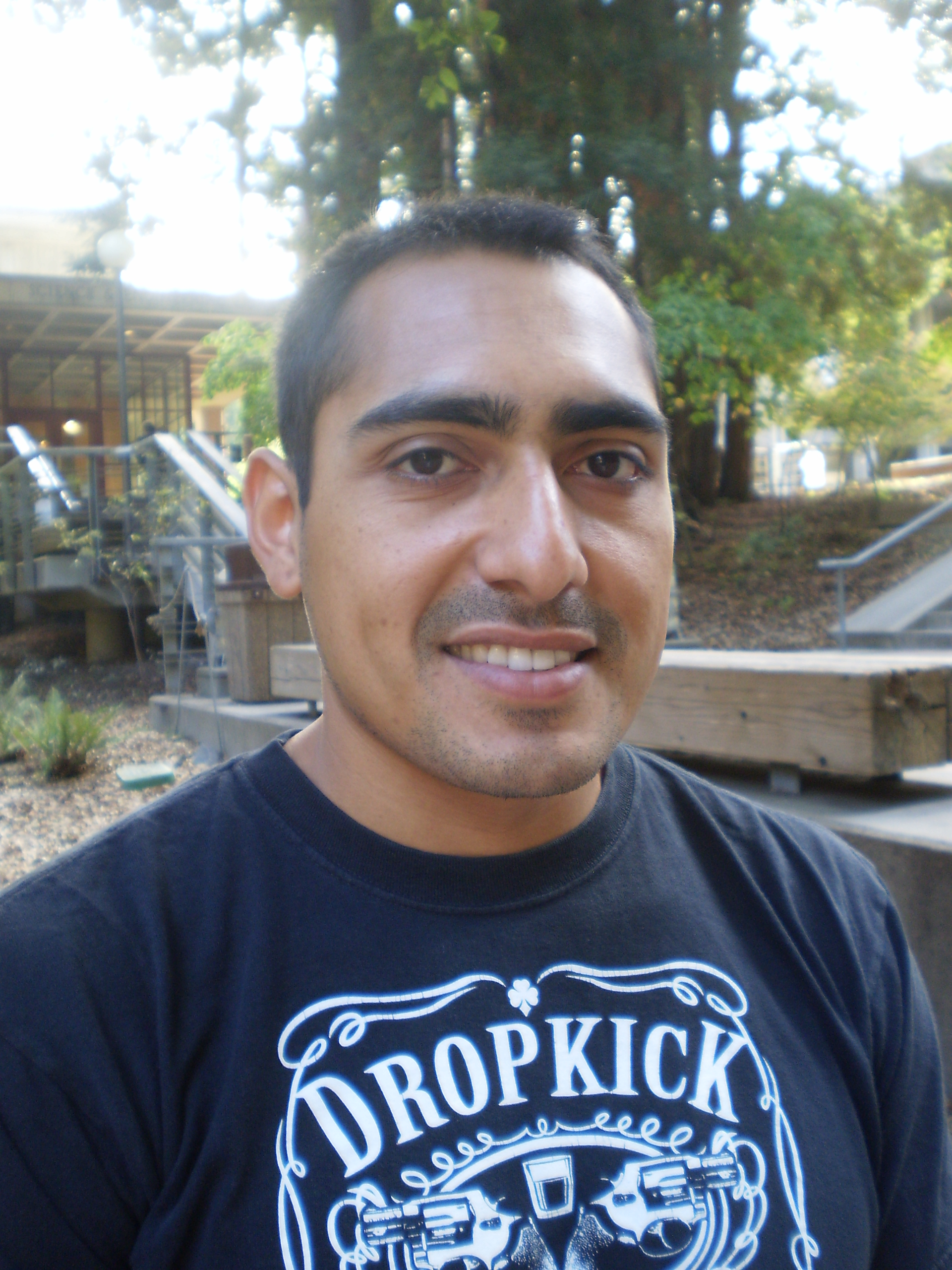}}]{Jos\'{e} Armando Oviedo}
Jos\'{e} Armando Oviedo received the B.S. in Electrical Engineering in 2009 from the California State Polytechnic University, Pomona, and the M.S. in Electrical Engineering in 2010 from the University of California, Riverside. He is currently pursuing the Ph.D. in Electrical Engineering at the University of California, Santa Cruz. 

His main areas of interest are non-orthogonal multiple access and multi-user diversity in multi-user wireless communication systems. 
\end{IEEEbiography}

\begin{IEEEbiography}[{\includegraphics[width=1in,height=1.25in,clip,keepaspectratio]{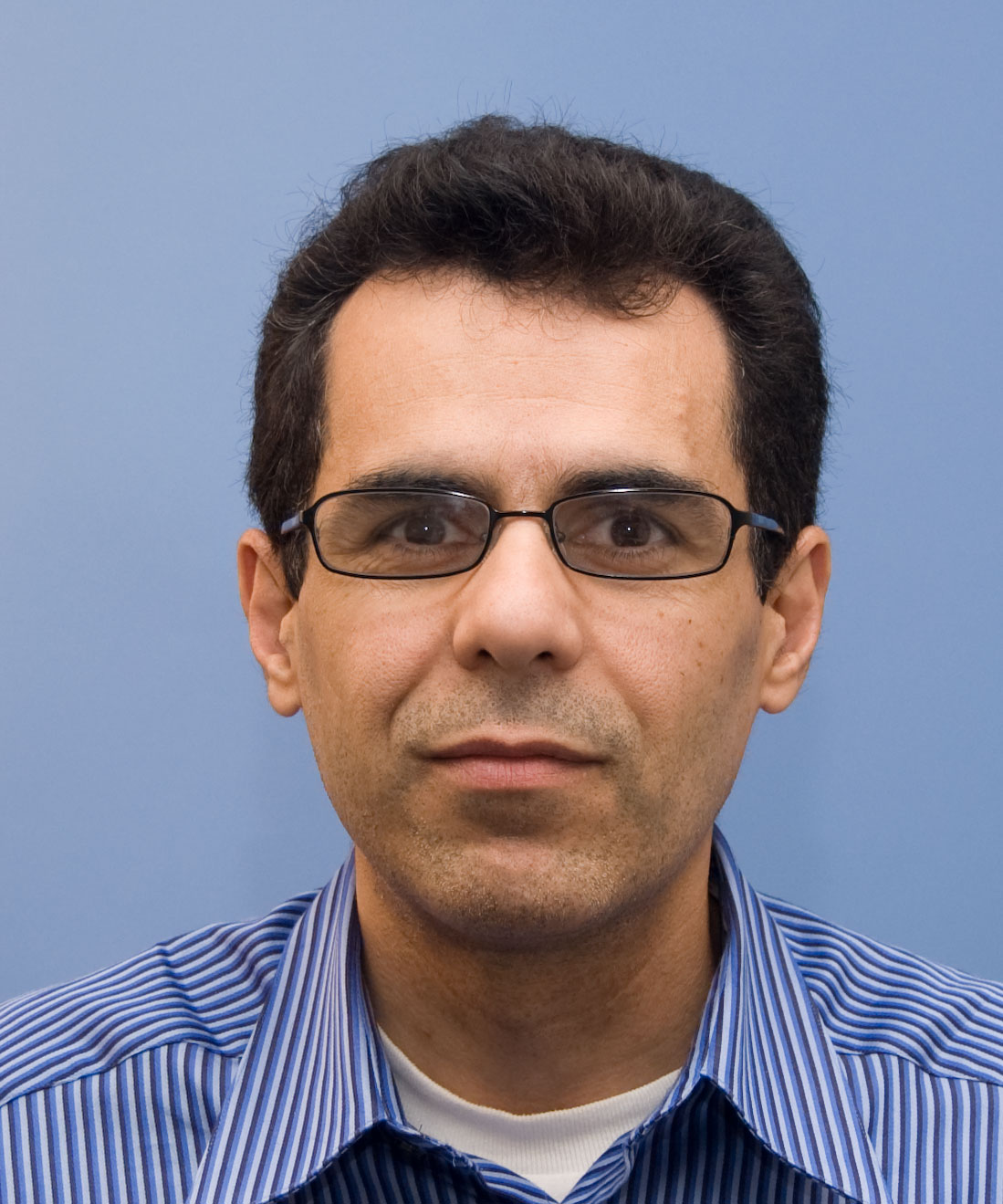}}]{Hamid R. Sadjadpour}
Hamid R. Sadjadpour (S’94–M’95–SM’00) received the B.S. and M.S. degrees from the Sharif University of Technology, and the Ph.D. degree from the University of Southern California at Los Angeles, Los Angeles, CA. In 1995, he joined the AT\&T Research Laboratory, Florham Park, NJ, USA, as a Technical Staff Member and later as a Principal Member of Technical Staff. In 2001, he joined the University of California at Santa Cruz, Santa Cruz, where he is currently a Professor. He has authored over 170 publications. He holds 17 patents. His research interests are in the general areas of wireless communications and networks. He has served as a Technical Program Committee Member and the Chair in numerous conferences. He is a co-recipient of the best paper awards at the 2007 International Symposium on Performance Evaluation of Computer and Telecommunication Systems and the 2008 Military Communications conference, and the 2010 European Wireless Conference Best Student Paper Award. He has been a Guest Editor of EURASIP in 2003 and 2006. He was
a member of the Editorial Board of Wireless Communications and Mobile Computing Journal (Wiley), and the  Journal OF Communications and Networks.
\end{IEEEbiography}

\end{document}